\documentclass[journal]{IEEEtran}   
\usepackage{multirow} 
\usepackage{diagbox} 
\usepackage{amssymb} 
\usepackage{amsmath} 
\usepackage{graphicx} 
\usepackage{cite} 
\usepackage{citesort}
\usepackage{balance} 
\usepackage[utf8]{inputenc} 
\usepackage{subcaption} 
\usepackage{epstopdf} 
\usepackage{comment} 
\usepackage{url} 
\usepackage{xcolor} 
\usepackage{amsthm} 
\usepackage{mdframed} 
\usepackage{algorithm} 
\usepackage{algpseudocode} 
\usepackage{bm} 
\usepackage{booktabs} 
\allowdisplaybreaks 

\newtheorem{theorem}{Theorem} 
\newtheorem{corollary}{Corollary} 
\newtheorem{remark}{Remark} 

\begin{document}
\title{\Huge{Performance Analysis of Pinching Antenna Systems Enabled  NOMA Communications}}

\author{ 
Xinwei~Yue,~\IEEEmembership{Senior Member,~IEEE}, Xinglun Tao, Jingjing Zhao,~\IEEEmembership{Senior Member,~IEEE}, Xianfu Lei, ~\IEEEmembership{Member,~IEEE}, Yuanwei\ Liu,~\IEEEmembership{Fellow, IEEE}, Zhiguo\ Ding,~\IEEEmembership{Fellow, IEEE}

\thanks{X. Yue and X. Tao are with the Center for Target Cognition Information Processing Science and Technology, Beijing Information Science $\&$ Technology University,and also with the Key Laboratory of Modern Measurement $\&$ Control Technology, Ministry of Education, Beijing Information Science $\&$ Technology University, Beijing 102206, China (email: \{xinwei.yue and xinglun.tao\}@bistu.edu.cn).}
\thanks{J. Zhao is with the School of Electronic and Information Engineering, Beihang University, Beijing 100191, China (email: jingjingzhao@buaa.edu.cn).}
\thanks{X. Lei is with the School of Information Science and Technology, Southwest Jiaotong University, Chengdu 611756, China (e-mail: xflei@swjtu.edu.cn).}
\thanks{Y. Liu is with the Department of Electrical and Electronic Engineering, The University of Hong Kong, Pokfulam, Hong Kong (email: yuanwei@hku.hk).}
\thanks{Z. Ding is with the School of Electrical and Electronic Engineering (EEE), Nanyang Technological University, Singapore 639798, (e-mail: zhiguo.ding@ntu.edu.sg).}
}
    
\maketitle

\begin{abstract}
Pinching antenna systems (PASS) have the advantages in the perspective of flexible antenna reconfiguration, line-of-sight (LoS) creation, and scalability features. To highlight the ascendancy of PASS, we survey the integration of PASS into non-orthogonal multiple access (NOMA) networks. The locations of nodes are randomly distributed within a circular coverage region. The influencing factors of LoS and none-line-of-sight (NLoS) propagation links from PASS to non-orthogonal nodes are taken into consideration. 
To characterize performance of PASS-NOMA, we deduce the outage probability and ergodic data rates expressions of two nodes over LoS/NLoS fading channels. In light of these theoretical results, the formally unbounded diversity gains are also characterized with node $n$ under ideal successive interference cancellation (ISIC) and node $f$ over LoS links, when the SNR exceeds a finite threshold. The slopes of ergodic data rate for node $n$ with ISIC equals to one and that of node $f$ presents a zero slope. In addition, the PASS-NOMA system throughput are evaluated in different transmission modes. 
It is shown from the numerical results that: 1) The outage behaviors of PASS-NOMA networks under LoS/NLoS conditions outperform that of PASS aided orthogonal multiple access (OMA) networks and that of conventional antenna systems (CASS) aided OMA/NOMA networks; 2) The employment of PASS enables the larger ergodic data rates relative to PASS-OMA and CASS-OMA/NOMA networks; and 3) As the quantity of communication radius and PASS' height rises, the performance of PASS-NOMA networks are declined at a certain extent.
\end{abstract}
\begin{IEEEkeywords}
Pinching antenna systems, non-orthogonal multiple access, outage probability, ergodic data rate, stochastic geometry.
\end{IEEEkeywords}

\section{Introduction}
The future sixth-generation communication techniques are anticipated to enable intelligent service awareness in terms of wide coverage, ultra-high data rate and allow users to connect each other anywhere \cite{Wang2023Roadto6G,Chen2023Toward6G}. To fulfill these ambitious goals, researchers have sought revolutionary paradigm, i.e., flexible antenna techniques \cite{Yuanwei2021RIS,Wong2021FAS,Zhu2024MAOpportunity}, which has ability to overcome the physical limitations of conventional wireless infrastructures in terms of reconfiguring the wireless channel and boosting its performance. Pinching antenna systems (PASS) have been regarded as one of key technologies in flexible antenna systems and its performance advantages were verified in \cite{NTT2022PinchingAntenna2022,Liu2025OutlookPASS}. The basic architecture of PASS is composed of dielectric waveguides and specialized dielectric particles, in which electromagnetic waves can be radiated out or collected through the pinched particles \cite{NTT2022PinchingAntenna}. The dielectric waveguide can support one or more activated radiation/collection points for transmitting or receiving the desired signals.

The integration of PASS into wireless communication networks was firstly studied in \cite{Ding2025FASsPASS}, demonstrating its validity in alleviating the impact of large-scale path loss on communication performance. A comprehensive tutorial on PASS was provided by systematically detailing its physical structure and control mechanisms \cite{Liu2025PASSTutorial}. To emphasize the implementation flexibility, the authors of \cite{Yang2025PAPrinciples} presented a broad overview of PASS principles and its application directions. An investigation into the array gain achievable by PASS configurations was presented in \cite{Ouyang2025ArrayGainPASS}, which also demonstrates that strategically activating pinching points can boost signal strength. In \cite{Wang2025ModelingPASS}, the beamforming optimization strategies for PASS were evaluated with consideration of the underlying hardware physics, addressing practical deployment concerns. Both distributed and centralized PASS deployment schemes were discussed in \cite{Xidong2025PASSCommunications}, where the distributed PASS deployment is capable of achieving maximum spectral efficiency. The performance analyses of PASS were highlighted with a focus on outage probability and average rate by taking account of line-of-sight (LoS) propagation links \cite{Tyrovolas2025PerformancePASS}. Moreover, the downlink data rates of PASS were maximized through optimization of the pinching antenna location \cite{Ding2025PASSdownlink}. Under realistic channel conditions, the uplink transmission performance of PASS was examined in \cite{Hou2025PerformancePASS}, evaluating both outage behavior and sum rates. Furthermore, the minimum data rates of uplink PASS were improved through dynamically adjust the placement of pinching antenna \cite{Tegos2025MinimumPASS}. The impact of the in-waveguide attenuation on PASS was analysed in \cite{Yanqing2025PASSPerformanceAnalysis}, where the corresponding degradation was shown to be significant mainly for long waveguides. Beyond performance characterization, recent efforts in \cite{Zhao2025PASSMUC} have also explored diverse transmission architectures for multi-user PASS, with the optimizing of joint beamforming and signal processing.

The above studies have significantly expanded our understanding of PASS architectures and their operational principles. The combinations of PASS with other wireless technologies have been furnished in \cite{Shan2025PASSMulticast,Yixuan2025PASSWPC,Jiang2025PASSCovert,Yaoyu2025PASSPositioning}. The PASS enabled multicast framework was developed in \cite{Shan2025PASSMulticast}, where the multicast rate is maximized in the cases of single dielectric waveguide and multiple dielectric waveguides. With the emphasis on green communications \cite{Yixuan2025PASSWPC}, the sum rate of PASS aided wireless powered communications was studied by invoking nonlinear energy harvesting model. For secure transmission requirement, the weighted secrecy throughput of PASS was evaluated in \cite{Mingjun2025PASSPLS} under the presence of multiple eavesdroppers.   
As a development, the authors of \cite{Jiang2025PASSCovert} shed light on how to implement PASS to enhance the covert performance. In \cite{Yaoyu2025PASSPositioning}, the indoor positioning performance achieved by PASS was revealed, in which the PASS is applied to estimate the users' multidimensional coordinates. Applying deep learning approaches, the authors of \cite{Xu2025PASSOptimization} explored the joint transmit and antenna placement of PASS enabled multi-user systems. A novel PASS assisted environment multiple access approach was discussed in \cite{Zhiguo2025PASSEDMA}, which effectively restrain interference by using signal blockages while enhancing the signal strength at intended users. In addition, the authors of \cite{Zhiguo2025PASSISAC} further analyzed the Caram\'{e}r-Rao bound of PASS based integrated sensing and communications. The flexibility of PASS was exploited for integrated sensing and communication in \cite{Khalili2025GlobecomISAC}, where activating antennas provides spatially diverse observations and target detection reliability.

The growing research interests have been dedicated to examine flexible antenna techniques assisted 
non-orthogonal multiple access (NOMA) communications, which serves as one of the pivotal enablers for next-generation wireless communication \cite{ITU2023NOMA,Liu2024NGMA50Year}. In \cite{Ding2020irsnoma}, the authors provided a straightforward design of reconfigurable intelligent surface (RIS) aided NOMA networks using simple control methods. The user fairness and sum rate of NOMA were maximized via the assistance of RIS partition \cite{Basar2022NOMARIS}. To solve the multiplicative fading problems, the authors of \cite{Yue2024ARISNOMA} introduced active RIS into NOMA networks, accounting for hardware impairments. Moving forward a single step, coverage area maximization in NOMA networks enabled by simultaneously transmitting and reflecting surfaces (STARS) was discussed in \cite{Xidong2022STARSNOMA}. On this basis, the application of active STARS into NOMA networks was furnished in \cite{Yue2024ASTARSNOMA}, where a detailed analysis of the outage behaviors and system throughput was presented. Driven by these trends, the authors of \cite{Kiat2024FASNOMA} applied the inimitable ability of fluid antenna system (FAS) to NOMA networks, which reveals how it enhances network performance relative to traditional antenna system. In \cite{Farshad2025FASNOMAISAC}, the integrated backscatter and sensing performance of FAS-NOMA was evaluated by making use of Gaussian copula. Leveraging the fluid paradigm, the authors of \cite{Yu2025FASSTARSNOMA} investigated the weighted sum rate of fluid STARS aided NOMA networks. 
As another parallel study, the effective throughput of movable antenna (MA) enabled NOMA short-packet transmission was studied in \cite{Qingqing2024MovableNOMA}. Considering hippopotamus optimization, the authors of \cite{Zhenyu2025MovableNOMA} investigated achievable rate maximization in MA-NOMA networks.

Building upon its integration with flexible antennas, recent efforts have extended this trend toward PASS to enable more adaptive and spatially flexible user access \cite{Ding2025FASsPASS,Qiao2025PASSNOMAorOMA,Chen2025PASSOMANOMA}. In \cite{Ding2025FASsPASS}, the utilization of PASS into NOMA networks was studied for the first time to verify the advantage of PASS in enhancing the ergodic sum rate. To highlight the benefits of non-orthogonal signaling, the authors of \cite{Qiao2025PASSNOMAorOMA} showed that PASS-NOMA networks achieve significantly larger communication coverage than their OMA-based counterparts. 
Under the special condition that the user and pinching antenna are kept in parallel positions, the authors in \cite{Chen2025PASSOMANOMA} discussed the user outage performance over LoS propagation for PASS-NOMA networks. An efficient power allocation mechanism for PASS-NOMA was developed in \cite{Saeed2025NOMAPASS} to minimize the total transmit power while satisfying individual users’ quality-of-service requirements. 
Moreover, the authors of \cite{WangKaidi2025PASSNOMA} focused on maximizing the sum rate in PASS-NOMA by jointly optimizing the locations and number of activated pinching antennas. Based on this work, the system energy consumption of PASS-NOMA was analyzed in \cite{Fu2025NOMAPASS} by taking into consideration multiple dielectric waveguides. With a focus on user fairness, the authors in \cite{Xu2025QoSNOMAPASS} further studied the resource allocation problem of PASS-NOMA by optimizing excitation point selection. A new PASS-NOMA framework called waveguide division multiple access scheme was recently proposed in \cite{Zhao2025PASSWDMA}, which assigns dedicated dielectric waveguides to individual users with the pinching beamforming.

\subsection{Motivations and Contributions}
Despite the essential foundation laid by the above works for comprehending PASS in wireless communications, the studies on integrating PASS and NOMA technologies are far from being well understood.
To be more specific, the authors in \cite{Ding2025FASsPASS} only evaluated the sum rate of PASS-NOMA, while the impact of LoS blockage on system performance remains unclear. Immediately afterwards, the outage behaviors of PASS-NOMA were effectively quantified under the special condition of keeping the node and pinching antenna in parallel positions \cite{Chen2025PASSOMANOMA}. Building on these contributions, we characterize the performance of PASS-NOMA networks in terms of outage probability, ergodic data rate, and system throughput, with a couple of non-orthogonal nodes placed within a circular communication region using stochastic geometry. The impact of cell radius, NLoS links and non-ideal successive interference cancellation (NISIC) on the network performance of PASS-NOMA was highlighted from the perspective of realistic requirement.   
Based on the above works, the core contributions of this paper are listed below:
\begin{enumerate}
  \item  We first outline a PASS-enabled NOMA communication framework, in which two destination nodes are stochastically distributed inside a circular coverage service area. We analyze the outage behaviours of nodes in PASS-NOMA networks over LoS propagation by taking into consideration ideal successive interference cancellation (ISIC) and NISIC conditions, and further provide insights into the outage performance for distant node over NLoS links.
  \item We deduce the closed-form expressions of outage probability and ergodic data rates under ISIC/NISIC conditions for destination nodes. As the quantity of communication radius and PASS' height rises, the outage probability and ergodic performance are declined at a certain extent.
  \item We further acquire the diversity gain of destination nodes over LoS/NLoS propagation links under high signal-to-noise ratio (SNR) conditions based on these approximated results. We observe that the diversity gain of near node with ISIC is formally unbounded above a finite SNR threshold, while it is zero on the condition of NISIC. Drawing on the analytical findings, the high SNR slopes of near node with ISIC and distant node over LoS/NLoS links are one and zero, respectively.     
  \item  We additionally survive the system throughput of PASS-NOMA networks under both delay-constrained and latency-tolerant transmissions. We also characterize the impact of the cell coverage radius on throughput performance. Our results show that PASS-NOMA has ability to achieve the enhanced system throughput relative to PASS aided orthogonal multiple access (PASS-OMA) and conventional antenna systems (CASS) based OMA/NOMA networks .
\end{enumerate}
\subsection{Organization and Notations}
The organization of this paper is briefly outlined. The Section II elaborates on the system model of PASS-NOMA networks. In Section III, we derive analytical expressions for the outage probability of nodes. 
The ergodic data rates of PASS-NOMA networks are studied in the Section IV. Subsequently, numerical results and performance insights are presented in the Section V. The Section VI concludes the main works, which also outlines potential directions for future research. The associated proof processes are presented in detail to verify the conclusion.

The main notations used in this paper are summarized below. For a random variable $X$, ${f_X}\left(  \cdot  \right)$ and ${F_X}\left(  \cdot  \right)$ represent its probability density function (PDF) and cumulative distribution function (CDF), respectively. The operator $\mathbb{E}\{\cdot\}$ stands for mathematical expectation, and $|| \cdot ||$ denotes the Euclidean norm (2-norm) of a vector or matrix. In addition, $M$ denotes the order of Gauss-Chebyshev quadrature adopted to approximate the involved integrals, and its value used is given in Table~\ref{parameter}.

\section{System Model}\label{System Model}
\subsection{Description of Network Architecture}
Focusing our attention on a PASS enabled downlink NOMA communication scenarios given in Fig. \ref{System_Model}, where the nodes' position are stochastically distributed in a circular area ${\cal D}$ of radius ${ R_{{\cal D}}}$. 
More specifically, the transmitting end equipped with PASS sends the superimposed messages to two single-antenna destination nodes, in which the PASS has one dielectric waveguide and $K$ pinching antennas\footnote{It is noteworthy to mention that considering multiple dielectric waveguides and pinching antennas can further enrich the channel gains of non-orthogonal nodes, which will be part of our future research.}. 
$ K $  pinching antennas are simultaneously activated and fed with the identical signal due to the shared-waveguide constraint \cite{Saeed2025NOMAPASS}. Given the intractable coherent superposition across $ K $ spatially distributed antennas, we adopt a representative antenna approximation firstly used in performance studies of pinching antenna systems \cite{Ding2025FASsPASS} for the purpose of analytical tractability. The effective channel is modeled as a single equivalent pinching antenna. 
The PASS is set parallel to the horizontal plane, at a height of $d$ above the ground.
It can be easily adjusted to establish the LoS communication area in the zone surrounding the reference pinching antenna. Without restricting generality, a couple of destination nodes i.e., near node $n$ and distant node $f$ are taken into consideration, where node $n$ is randomly located within the circular area of radius $R_n$. Concurrently, node $f$ is uniformly distributed in the circular area denoted by ${ R_{f}}$. At this moment, the location of node $n$ is closer to the origin than that of node $f$, i.e., ${R_{n}}<{ R_{f}} \le{ R_{{\cal D}}}$.\footnote{The nodes are defined statistically through their distribution radii, i.e., $R_n<R_f$. Despite a realization in which node $f$ lies closer to the origin than node $n$ is possible, it is covered by averaging over the above distributions in subsequent analysis.} Since node $f$ is far from the the transmitting end, it is particularly vulnerable to LoS blockages compared to nodes in closer proximity. The performance of node $f$ under NLoS channel conditions in PASS-NOMA networks will be studied in the following section. To obtain straightforward results, we postulate that the transmitting end has the destination nodes' ideal channel state information (CSI).
\begin{figure}
  \centering
  \includegraphics[width= 3.48in, height=2.0in]{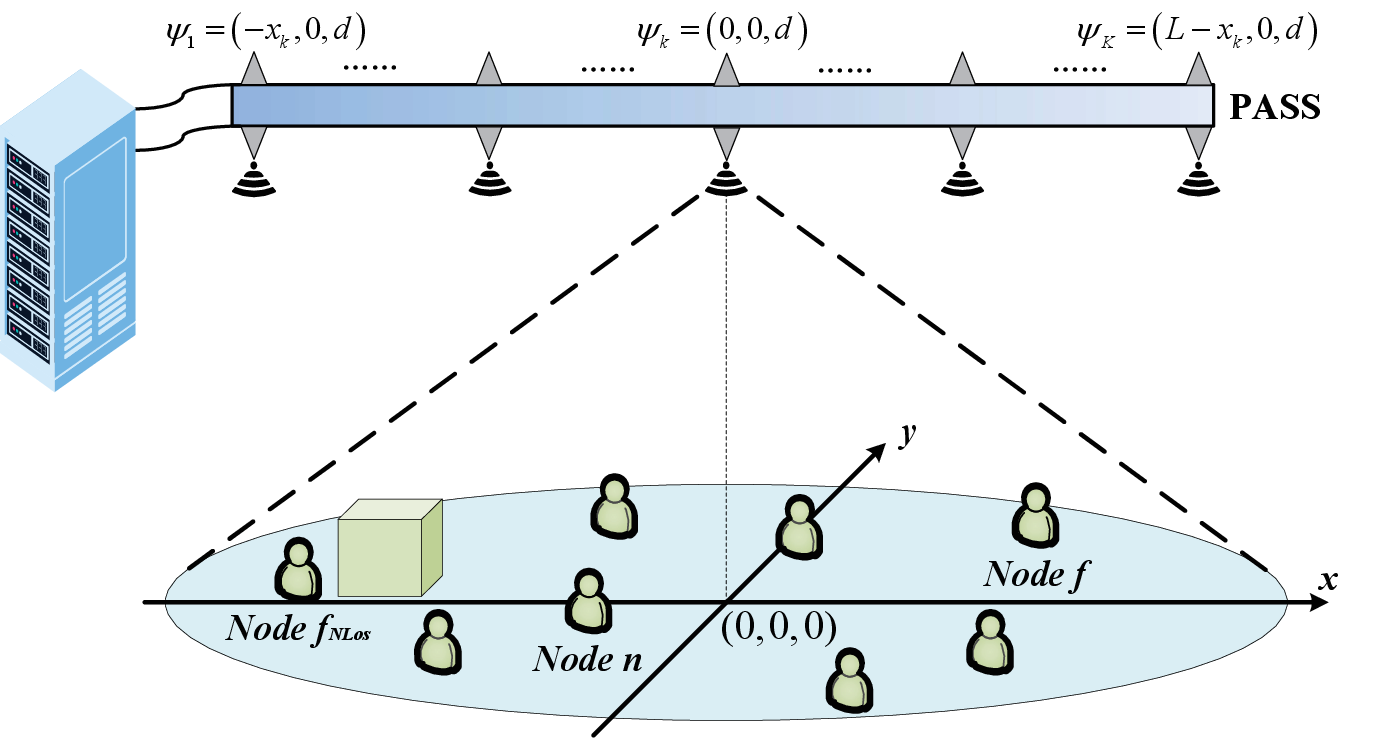}
  \caption{A schematic diagram of PASS-NOMA networks with spatially random nodes.}
  \label{System_Model}
\end{figure}

\subsection{Signal Model}
The transmitting end sends the superposed signal via all $K$ activated pinching antennas on the waveguide to a couple of non-aligned nodes. The received signal of node $\varphi$, i.e., $\varphi  \in \{ n,f\} $ can be given by
\begin{align}\label{received signal of user varphi}
y_\varphi = \left( \sum_{k=1}^K h_{\varphi,k} \phi _{k} \right) \left( \sqrt{\frac{a_n P_b}{K}} s_n + \sqrt{\frac{a_f P_b}{K}} s_f \right) + w_\varphi,
\end{align}
where ${s_\varphi }$ denotes the energy normalized signal of node $\varphi$ with $\mathbb{E}\{s_\varphi^2 \}= 1$. 
${a_n}$ and ${a_f}$ stand for the power assignment parameters for node $n$ and node $f$ with the condition of ${a_n} < {a_f}$. 
$P_b$ represents the transmitting end's power and $w_{\varphi} \sim  {\cal C}{\cal N} \left(0, \sigma^2 \right)$ denotes the additive Gaussian white noise at node $\varphi$. 
Assuming that the transmitting end's power is homogeneously assigned to the $K$ pinching antennas, the transmit power of one of the activated antennas can then be expressed as $\frac{{{P_b}}}{K}$\footnote{Note that it is important to focus on the different power allocated configuration at the activated antennas \cite{Khalili2025GlobecomISAC}, which will be set aside in our future works.}. $h_{\varphi,k}$ denotes the propagation channel link from the $k$-th pinching antenna to node $ \varphi $. 
Due to the random placement of antennas along the waveguide, the distribution of $\sum_{k=1}^K h_{\varphi,k} {\phi _{k}}$ is analytically intractable. Following the common practice in \cite{Ding2025FASsPASS}, we approximate the $k$-th pinching antenna as the representative antenna by ${h_\varphi } \approx \frac{{ \sqrt{\varepsilon \eta} {\tilde{h}}_\varphi  {e^{ - j\frac{{2\pi }}{\lambda }\left\| {{{\bf{u}}_\varphi } - {{\boldsymbol{\psi }}_k}} \right\|}}}}{\left\| {{\bf{u}}_\varphi } - {{\boldsymbol{\psi }}_k} \right\|^{\frac{\alpha}{2}}}$. It reduces to the LoS channel with ${{\tilde{h}}_\varphi }=1$ and $\alpha=2$. $\alpha$ denotes the path loss factor. $\varepsilon$ denotes the equivalent power coefficient that characterizes the aggregate channel gain of $K$ activated pinching antennas relative to a single equivalent antenna, thereby extending the equivalent model to the waveguide-distributed antenna configuration. The $K$ activated pinching antennas are equally spaced along the dielectric waveguide, which is centered at the equivalent antenna. $\eta  = \frac{{{c^2}}}{{16{\pi ^2}f_c^2}}$, $c$ and $f_c$ stand for the speed of light and carrier frequency, correspondingly. $\lambda $ denotes the wavelength of transmit signal in free space. 
The $k$-th pinching antenna is located at $\boldsymbol{\psi}_k = \left(x_k, 0, d \right)$ along the waveguide, where $x_k$ represents the propagation distance from the $k$-th pinching antenna to the transmitting end. $\mathbf{u}_{\varphi} = \left( x_{\varphi}, y_{\varphi}, 0 \right)$ denotes the random location of node $\varphi$ in the desired communication area, which is uniformly distributed in the circular communication area $\mathcal{D} = \{(x,y,0) : x^2 + y^2 \leq R_{\varphi}^2\}$. 
The phase shift i.e., ${\phi _{k}} = {e^{-j\frac{{2\pi {x_k}}}{{{\lambda _g}}}}}$ of acquired signal transmitted via the waveguide from the $k$-th pinching antenna, where $\lambda_g = \frac{\lambda}{n_{eff}}$ is the wavelength of the signal transmitted, with $n_{eff}$ being the effective refractive index of PASS. As stated in \cite{Yanqing2025PASSPerformanceAnalysis}, the in-waveguide attenuation can be safely neglected over the communication regions in this paper.

By employing SIC scheme, node $n$ first detect node $f$'s signal $s_f$ and its signal-to-interference-plus-noise ratio (SINR) can be described as
\begin{align}\label{SINR n to f}
{\gamma _{n \to f}} = \frac{{\rho {{\left| {{h_n}{\phi _k}} \right|}^2}{a_f}}}{{\rho {{\left| {{h_n}{\phi _k}} \right|}^2}{a_n} + 1}},
\end{align}
where $\rho  = \frac{{{P_b}}}{{K{\sigma ^2}}}$ denotes the transmit SNR with the total power $P_b$ equally shared by the $K$ activated pinching antennas, and the condition ${a_f} > {\gamma _{thf}}{a_n}$ ensures that node $n$ can decode ${s_f}$ before performing SIC. After removing node $n$'s information using SIC, the SINR received at node $n$ for decoding its own signal $s_n$ is expressed as follows:
\begin{align}\label{SINR n}
{\gamma _n} = \frac{{\rho {{\left| {{h_n}{\phi _k}} \right|}^2}{a_n}}}{{\varpi \rho {{\left| {{h_I}} \right|}^2} + 1}},
\end{align}
where $\varpi$ is the conversion coefficient. To be more specific, when $\varpi$ is set to be $\varpi = 1$, there is leftover interference from the NISIC carried out at node $n$. On the contrary, $\varpi = 0$ represents that node $n$ makes use of the ISIC to detect its information. ${h_I} \sim {\cal C}{\cal N} (0,{\Omega _I})$ can be modeled as the leftover interference with the NISIC.

In addition, the SINR received at node $f$ to detect the signal $s_f$ by considering $s_n$ as interference can be given by 
\begin{align}\label{SINR f}
{\gamma _f} = \frac{{\rho {{\left| {{h_f}{\phi _k}} \right|}^2}{a_f}}}{{\rho {{\left| {{h_f}{\phi _k}} \right|}^2}{a_n} + 1}}.
\end{align}

\section{Outage Probability}\label{Outage Probability}
In this section, the outage probability of node $n$ under ISIC/NISIC and node $f$ in PASS-NOMA networks is analyzed in depth. The asymptotic outage probability and diversity gain for node $\varphi$ are achieved at high SNRs to characterize the performance.  
\subsection{Outage Probability of Node $n$}\label{subsection OP of user n}
In accordance with the NOMA protocol, the outage probability event of node $n$ is defined as follows: 1) If node $n$ fails to decode $s_f$, the outage event occurs; 2) Although the signal $s_f$ is decoded successfully, node $n$ cannot detected its information $s_n$ and the communication may also be interrupted. Based on these explanations, the outage probability of node $n$ over LoS propagation links in PASS-NOMA networks is usually defined by
\begin{align}\label{OP of user n defined}
{{\rm{P}}_n^{{\rm{LoS}}}} =  {\rm{P}}({\gamma _{n \to f}} > {\gamma _{thf}},{\gamma _n} < {\gamma _{thn}}) +{\rm{P}}({\gamma _{n \to f}} < {\gamma _{thf}}) , 
\end{align}
where 
${\gamma _{thn}} = 2^{{\hat{R}}_n} - 1$ and ${\gamma _{thf}} = 2^{{\hat{R}}_f} - 1$ represent the target SNRs required to decode $s_n$ and $s_f$. The associated target data rates for nodes $n$ and $f$ are given by ${\hat{R}}_n$ and ${\hat{R}}_f$, respectively. The outage probability expression of node $n$ under NISIC in PASS-NOMA networks is given in the theorem below.
\begin{theorem}\label{theorem OP expression of user n with ipSIC}
The closed-form outage probability expression of node $n$ under NISIC over LoS propagation links can be expressed as
\begin{align}\label{OP of user n with ipSIC}
{\rm{P}}_{n,{\rm{NISIC}}}^{{\rm{LoS}}} = \begin{cases}
1, & {{C_n} < {d^2}}, \\
\begin{aligned}[b] & 1 + {\xi_4} - {\xi_4}{e^{ - \frac{{{\xi _1}}}{{{d^2}}} + \frac{1}{{{\Omega _I}\rho }}}} \\ & - \frac{{{\xi _1}}}{{R_n^2}}{e^{\frac{1}{{{\Omega _I}\rho }}}}{\xi _2}, \end{aligned} & {{d^2} < {C_n} < R_n^2 + {d^2}}, \\
\begin{aligned}[b] & {e^{\frac{1}{{{\Omega _I}\rho }}}}\left[ {\left( {1 + {\xi_4}} \right){e^{ - \frac{{{\xi _1}}}{{(R_n^2 + {d^2})}}}}} \right. \\ & \left. { - {\xi_4}{e^{ - \frac{{{\xi _1}}}{{{d^2}}}}} - \frac{{{\xi _1}}}{{R_n^2}}{\xi _3}} \right], \end{aligned} & {{C_n} > R_n^2 + {d^2}}, \\
\end{cases}
\end{align}
where 
$C_n=\frac{\varepsilon\eta\rho a_n}{\gamma_{thn}}$, 
${\xi _1} = \frac{{\varepsilon \eta {a_n}}}{{{\Omega _I}{\gamma _{thn}}}},{\xi _2} = {\rm{Ei}}\left( { - \frac{{{\xi _1}}}{{{d^2}}}} \right) - {\rm{Ei}}\left( { - \frac{1}{{{\Omega _I}\rho }}} \right),{\xi _3} = {\rm{Ei}}\left( { - \frac{{{\xi _1}}}{{{d^2}}}} \right) - {\rm{Ei}}\left( { - \frac{{{\xi _1}}}{{R_n^2 + {d^2}}}} \right),{\xi _4} = \frac{{{d^2}}}{{R_n^2}}$,
$\rm{Ei}\left( \cdot \right)$ is the exponential integral function\cite[Eq. (8.211.1)]{2000gradshteyn}.
\begin{proof}
See the following Appendix~A.~\ref{Appendix:A}
\end{proof}
\end{theorem}
\begin{corollary}\label{corollary OP of user n with pSIC}
When $\varpi = 0$ is satisfied, the closed-form outage probability expression of node $n$ under ISIC in PASS-NOMA networks can be obtained as
\begin{align}\label{OP of user n with pSIC}
{\rm{P}}_{n,{\rm{ISIC}}}^{{\rm{LoS}}} = \begin{cases}
1, & {C_n} < {d^2}, \\
1 - \frac{{{C_n} - {d^2}}}{{R_n^2}}, & {d^2} < {C_n} < R_n^2 + {d^2}, \\
0, & {C_n} > R_n^2 + {d^2}.
\end{cases}
\end{align}
\end{corollary}

\subsection{Outage Probability of Node $f$}\label{subsection OP of user f}
With $s_n$ regarded as the interference, node $f$ cannot decode its own signal $s_f$ and an outage occurs. At this point, the corresponding outage probability of node $f$ via LoS/NLoS propagation links in PASS-NOMA networks is expressed by
\begin{align}\label{OP of user f defined}
{{\rm{P}}_f} = P\left( {{\gamma _f} < {\gamma _{thf}}} \right).
\end{align}

\begin{theorem}\label{theorem OP of user f}
The closed-form outage probability expression of node $f$ over LoS propagation links in PASS-NOMA networks can be given by
\begin{align}\label{OP of user f}
{{\rm{P}}_f^{{\rm{LoS}}}} = \begin{cases}
1, & {{C_f} < {d^2}}, \\
{1 - \frac{{{C_f} - {d^2}}}{{R_f^2}}}, & {{d^2} < {C_f} < R_f^2 + {d^2}}, \\
0, & {{C_f} > R_f^2 + {d^2}},
\end{cases}
\end{align}
where 
${C_f} = \varepsilon \eta \rho (\frac{{{a_f}}}{{{\gamma _{thf}}}} - {a_n})$ on the condition of ${a_f} > {\gamma _{thf}}{a_n}$.
\begin{proof}
Upon plugging \eqref{SINR f} into \eqref{OP of user f defined} and through some basic manipulations, the outage probability of node $f$ is further reckoned as
\begin{align}\label{the derivation of OP of user f}
{{\rm{P}}_f^{{\rm{LoS}}}} &= {\mathrm{P}}\left( {{x_f^2} + y_f^2 > \varepsilon \eta \rho (\frac{{{a_f}}}{{{\gamma _{thf}}}} - {a_n}) - {d^2}} \right) \nonumber \\
&= 1 - F_{r^2_f}\left[ \varepsilon \eta \rho (\frac{{{a_f}}}{{{\gamma _{thf}}}} - {a_n}) - {d^2} \right].
\end{align}
We further substitute \eqref{The CDF of quare of distance} into the above equation to derive \eqref{OP of user f}. This completes the proof.
\end{proof}
\end{theorem}

Since node $ f $ is randomly located within the cell, it may reside at the cell edge or behind obstacles, leading to complete blockage of the LoS path from the PASS. In such scenarios, communication occurs solely via NLoS multipath components. Without a dominant propagation path, this multipath fading is well modeled as Rayleigh distribution, i.e.,  $ {\tilde{h}}_f \sim \mathcal{CN}\left( 0, \Omega_f \right) $ , where  $ \Omega_f $  denotes the average channel power.
The following corollary provides the theoretical result.
\begin{corollary}\label{corollary OP of user f with NLoS}
The closed-form outage probability expression of node $f$ over NLoS propagation links in PASS-NOMA networks is given by
\begin{align}\label{OP of user f with NLoS}
{\rm{P}}_f^{{\rm{NLoS}}} = 1 - \frac{{{\Omega _f}{C_f}}}{{R_f^2}}\left( {{e^{ - \frac{{{d^2}}}{{{\Omega _f}{C_f}}}}} - {e^{ - \frac{{R_f^2 + {d^2}}}{{{\Omega _f}{C_f}}}}}} \right).
\end{align}
\begin{proof}
See the following Appendix~B.~\ref{Appendix:B}
\end{proof}
\end{corollary}

\subsection{Diversity Gain Analyses}\label{Diversity Analysis}
The diversity gain characterizes the high SNR outage behavior and reflects a system’s resilience to channel variations \cite{David2003Order,Ding6868214}. In PASS, however, all pinching antennas share a single waveguide and transmit identical signals, making the effective channel dominated by large scale propagation path loss rather than small scale propagation links. As a result, the diversity performance is primarily determined by node location randomness and LoS/NLoS propagation rather than the number of antennas. The corresponding diversity gain expression is defined as
\begin{align}\label{definition expression of diversity order} 
{D_o} = - {\mathop {\lim} \limits_{\rho \to \infty} } \frac{\log \left[ P^{\infty} \left( \rho \right) \right]}{\log \left( \rho \right) },
\end{align}
where $P^{\infty}\left( \rho \right)$ expresses the asymptotic outage probability expression of nodes at high SNRs. 
\begin{corollary}\label{corollary asym OP of user n with ipSCI}
On the condition of $\rho \to \infty$, the asymptotic outage probability expression of node $n$ under NISIC can be written as
\begin{align}\label{asym OP of user n with ipSCI}
{\rm{P}}_{n,\infty }^{{\rm{NISIC}}} \approx & \left( {1 + \frac{{{d^2}}}{{R_n^2}}} \right){e^{ - \frac{{{\xi _1}}}{{R_n^2 + {d^2}}}}} - \frac{{{d^2}}}{{R_n^2}}{e^{ - \frac{{{\xi _1}}}{{{d^2}}}}} \nonumber \\
&- \frac{{{\xi _1}}}{{R_n^2}}\left[ {{\rm{Ei}}\left( { - \frac{{{\xi _1}}}{{{d^2}}}} \right) - {\rm{Ei}}\left( { - \frac{{{\xi _1}}}{{R_n^2 + {d^2}}}} \right)} \right].  
\end{align}
\begin{remark}\label{remark asym OP of user n with ipSCI}
Upon plugging \eqref{asym OP of user n with ipSCI} into \eqref{definition expression of diversity order}, the diversity gain of node $n$ under NISIC in PASS-NOMA networks equals zero, since the leftover interference remains non-vanishing even as the SNR grows unbounded.
\end{remark}
\end{corollary}

\begin{corollary}\label{corollary4 asym OP of user n with pSCI}
In the case of $\rho \to \infty$ and ${C_n} > R_n^2 + {d^2}$, the asymptotic outage probability expression of node $n$ under ISIC over LoS propagation links is given by
\begin{align}\label{asym OP of user n with pSCI}
{\rm{P}}^{\rm{ISIC}}_{n,\infty} = 0.
\end{align}
\begin{remark}\label{remark asym OP of user n with pSCI}
Upon plugging \eqref{asym OP of user n with pSCI} into \eqref{definition expression of diversity order}, 
the diversity gain of node $n$ under ISIC in PASS-NOMA networks formally approaches infinity in the sense of \eqref{definition expression of diversity order}. We emphasize that this is not a physical multipath diversity mechanism. It is a mathematical consequence of the deterministic LoS geometry and the bounded user region, for which the piecewise outage probability becomes exactly zero once the SNR exceeds a threshold.
\end{remark}
\end{corollary}

\begin{corollary}\label{corollary asym OP of user f with LoS}
On the condition of $\rho \to \infty$ and ${C_f} > R_f^2 + {d^2}$, the asymptotic outage probability expression of node $f$ over LoS propagation links in PASS-NOMA networks is obtained as
\begin{align}\label{asym OP of user f with LoS}
{\rm P}_{f, \infty}^{\rm LoS} = 0.
\end{align}
\begin{remark}\label{remark asym OP of user f with LoS}
Upon plugging \eqref{asym OP of user f with LoS} into \eqref{definition expression of diversity order}, 
the diversity gain of node $f$ with LoS propagation links in PASS-NOMA networks formally tends to infinity in the sense of \eqref{definition expression of diversity order}. We emphasize that this is not a physical multipath diversity mechanism, but it is a mathematical consequence of the deterministic LoS geometry and the bounded user region, for which the piecewise outage probability becomes exactly zero once the SNR exceeds a threshold.
\end{remark}
\end{corollary}

\begin{corollary}\label{corollary asym OP of user f with NLoS}
The asymptotic expression for outage probability of node $f$ over NLoS propagation links in PASS-NOMA networks can be given by
\begin{align}\label{asym OP of user f with NLoS}
{\rm{P}}_{f,\infty }^{{\rm{NLoS}}} \approx \frac{{2{d^2} + R_f^2}}{{2{\Omega _f}{C_f}}}.
\end{align}
\begin{remark}\label{remark asym OP of user f with NLoS}
Upon plugging \eqref{asym OP of user f with NLoS} into \eqref{definition expression of diversity order}, the diversity gain of node $f$ over NLoS propagation links in PASS-NOMA networks equals one. The main reason is that as the SNRs increase, the large-scale NLoS attenuation asymptotically dominates, and the outage probability exhibits first-order attenuation.
\end{remark}
\end{corollary}

\subsection{Delay-Constrained Throughput Analysis}\label{delay-limited mode System throughput}
The system throughput under delay-constrained transmission depends on the outage probability associated with a given target data rate \cite{Ding6868214,Yue8026173}. Based on the previously outage probability expression derived, the delay-constrained system throughput of PASS-NOMA networks with ISIC/NISIC can be defined as 
\begin{align}\label{Definition System Throughput}
{T_{{dl}}} = \left[ {1 - {\rm{P}}_n^{{\nu _1}}\left( \rho  \right)} \right]{\hat{R}}_{n} + \left[ {1 - {\rm{P}}_f^{{\nu _2}}\left( \rho  \right)} \right]{\hat{R}}_{f},
\end{align}
where 
$\nu_1 \in \left\{ \rm{NISIC}, \rm{ISIC} \right\}$, $\nu_2 \in \left\{ \rm{LoS}, \rm{NLoS} \right\}$,
${\rm{P}}_n^{NISIC}$ and ${\rm{P}}_n^{ISIC}$ can be obtained from \eqref{OP of user n with ipSIC} and \eqref{OP of user n with pSIC}, respectively. ${\rm{P}}_n^{LoS}$ and ${\rm{P}}_n^{NLoS}$ can be acquired from \eqref{OP of user f} and \eqref{OP of user f with NLoS}, respectively.

\section{Ergodic Data Rate}\label{Ergodic Data Rate}
In this section, we further evaluate the ergodic data rate of PASS-NOMA networks in terms of deriving the approximate expressions \cite{WangKaidi2025PASSNOMA}, which is expressed as
\begin{align}\label{Definition Ergodic Rate}
{\tilde R_\varphi } = \frac{1}{{\ln 2}}\int_0^\infty  {\frac{{1 - {F_{{\gamma _ \varphi }}}\left( x \right)}}{{1 + x}}{\rm{d}}x}.
\end{align}

\subsection{Ergodic Data Rate of Node $n$}
Plugging \eqref{SINR n} into \eqref{Definition Ergodic Rate}, the ergodic data rate of node $n$ in PASS-NOMA networks can be obtained as follows.
\begin{theorem}\label{theorem ER of user n with ipSIC}
The approximate expression of ergodic data rate for node $n$ under NISIC over LoS propagation links can be presented as
\begin{align}\label{ER of user n with ipSIC}
&\tilde R_n^{\rm{NISIC}} \approx \frac{\pi }{{M\ln 2}}\sum\limits_{m = 1}^M \left\{ \frac{\rho \varepsilon \eta a_n \sqrt{1-t_m^2}}{2\left(R_n^2 + d^2 \right)\left( 1 + \zeta_3 \right)} \right.  \nonumber \\ 
& \left. \times \left\{ 1 - e^{\frac{1}{\Omega_I \rho}} \left[ \left( 1 + \frac{d^2}{R_n^2} \right) e^{-\frac{\varepsilon \eta a_n}{\zeta_3 \Omega_I \left( R_n^2 + d^2 \right)}} - \frac{d^2}{R_n^2} e^{-\frac{\varepsilon \eta a_n}{\zeta_3 \Omega_I d^2}} \right. \right. \right. \nonumber \\ 
& \left. \left. \left.  - \frac{\varepsilon \eta a_n \zeta_1}{\zeta_3 \Omega_I R_n^2} \right] \right\} + \frac{\rho \varepsilon \eta a_n R_n^2 \sqrt{1-t_m^2}}{2 d^2 \left( R_n^2 + d^2 \right) \left( 1+ \zeta_4 \right)} \left[ -\frac{d^2}{R_n^2} \left( 1  \right. \right. \right. \nonumber \\ 
&\left. \left. \left. - e^{-\frac{\varepsilon \eta a_n}{\zeta_4 \Omega_I d^2} + \frac{1}{\Omega_I \rho}} \right) + \frac{\varepsilon \eta a_n \zeta_2}{\zeta_4 \Omega_I R_n^2} e^{\frac{1}{\Omega_I \rho}}   \right] \right\}, 
\end{align}
where ${t_m} = \cos \left( {\frac{{2m - 1}}{{2M}}\pi } \right),
{\zeta _1} = {\rm{Ei}}\left( { - \frac{{\varepsilon \eta {a_n}}}{{{\zeta _3}{\Omega _I}{d^2}}}} \right) - {\rm{Ei}}\left( { - \frac{{\varepsilon \eta {a_n}}}{{{\zeta _3}{\Omega _I}\left( {R_n^2 + {d^2}} \right)}}} \right),{\zeta _2} = {\rm{Ei}}\left( { - \frac{{\varepsilon \eta {a_n}}}{{{\zeta _4}{\Omega _I}{d^2}}}} \right) - {\rm{Ei}}\left( { - \frac{1}{{{\Omega _I}\rho }}} \right),
{\zeta _3} = \frac{{\rho \varepsilon \eta {a_n}\left( {{t_m} + 1} \right)}}{{2\left( {R_n^2 + {d^2}} \right)}},{\zeta _4} = \frac{{\rho \varepsilon \eta {a_n}R_n^2\left( {{t_m} + 1} \right)}}{{2{d^2}\left( {R_n^2 + {d^2}} \right)}} + \frac{{\rho \varepsilon \eta {a_n}}}{{R_n^2 + {d^2}}}$.
\begin{proof}
See Appendix~C.
\end{proof}
\end{theorem}
\begin{corollary}\label{corollary ER of user n with pSIC}
When $\varpi = 0$ is satisfied, the closed-form expression for ergodic data rate of node $n$ under ISIC over LoS propagation links can be written as
\begin{align}\label{ER of user n with pSIC}
&\tilde R_n^{{\rm{ISIC}}} = \frac{1}{{\ln 2}}\left[ \frac{{\varepsilon \eta \rho {a_n}}}{{R_n^2}}\ln \left( {\frac{{R_n^2 + {d^2}}}{{{d^2}}}} \right) - \frac{{\varepsilon \eta \rho {a_n} + {d^2}}}{{R_n^2}} \right.  \nonumber \\
& \left. \times \ln \left( {1 + \frac{{\varepsilon \eta \rho {a_n}}}{{{d^2}}}} \right) + \left( {\frac{{\varepsilon \eta \rho {a_n} + {d^2}}}{{R_n^2}} + 1} \right)\ln \left( {1 + \frac{{\varepsilon \eta \rho {a_n}}}{{R_n^2 + {d^2}}}} \right) \right].
\end{align}
\begin{proof}
See the following Appendix~D.
\end{proof}
\end{corollary}

\subsection{Ergodic Data Rate of Node $f$}
Following the same procedure as node $n$, we evaluate the ergodic data rate of node $f$ over both LoS and NLoS channels in PASS-NOMA networks. Upon plugging the CDF of \eqref{SINR f} into \eqref{Definition Ergodic Rate}, the ergodic data rate expressions of node $f$ are provided in the following contents.
\begin{theorem}\label{theorem ER of user f}
The closed-form ergodic data rate expression of node $f$ over LoS propagation links in PASS-NOMA networks is given by
\begin{align}\label{ER of user f}
{{\tilde R}_f^{\rm LoS}} &= {\log _2}\left( {1 + \frac{{\varepsilon \eta \rho {a_f}}}{{R_f^2 + {d^2} + \varepsilon \eta \rho {a_n}}}} \right) + \frac{{\varepsilon \eta \rho {a_f}}}{{R_f^2}}  \nonumber \\
& \times {\log _2}\left( {1 + \frac{{R_f^2}}{{{d^2} + \varepsilon \eta \rho {a_n}}}} \right) - \frac{{\varepsilon \eta \rho  + {d^2}}}{{R_f^2}}  \nonumber \\
& \times {\log _2}\left[ {\frac{{\left( {{d^2} + \varepsilon \eta \rho {a_n} + \varepsilon \eta \rho {a_f}} \right)\left( {R_f^2 + {d^2} + \varepsilon \eta \rho {a_n}} \right)}}{{\left( {{d^2} + \varepsilon \eta \rho {a_n}} \right)\left( {R_f^2 + {d^2} + \varepsilon \eta \rho {a_n} + \varepsilon \eta \rho {a_f}} \right)}}} \right].
\end{align}
\begin{proof}
See the following Appendix~E.
\end{proof}
\end{theorem}
\begin{corollary}\label{corollary ER of user f with NLoS}
The approximate ergodic data rate expression of node $f$ over NLoS propagation links in PASS-NOMA networks is presented as 
\begin{align}\label{ER of user f with NLoS}
 \tilde R_f^{{\rm{NLoS}}} \approx & {\rm{ }}\frac{\pi }{{M\ln 2}}\sum\limits_{m = 1}^M {\left\{ {\frac{{{a_f}\kappa \sqrt {1 - t_m^2} }}{{2{a_n}R_f^2\left[ {1 + \frac{{{a_f}}}{{2{a_n}}}\left( {t_m + 1} \right)} \right]}}} \right.} {\rm{ }}  \nonumber  \\ 
 & \left. { \times \left( {{e^{ - \frac{{{d^2}}}{\kappa }}} - {e^{ - \frac{{R_f^2 + {d^2}}}{\kappa }}}} \right)} \right\}, 
\end{align}
where ${t_m} = \cos \left( {\frac{{2m - 1}}{{2M}}\pi } \right),\kappa  = {\Omega _f}\rho \varepsilon \eta {a_n}\left( {\frac{2}{{t_m + 1}} - 1} \right)$.
\begin{proof}
See the following Appendix~F.
\end{proof}
\end{corollary}
\subsection{Slope Analyses}
The high SNR slope of ergodic data rate is one of the key parameters for evaluating the average performance of communication systems similar to the diversity gain shown above, which is usually defined as
\begin{align}\label{Definition high SNR Slope}
S = \mathop {\lim }\limits_{\rho  \to \infty } \frac{{\tilde R\left( \rho  \right)}}{{\log \left( \rho  \right)}},
\end{align}
where ${\tilde R\left( \rho  \right)}$ denotes the asymptotic expression for the ergodic data rate of nodes at high SNRs. 
Based on the above, the conventional high SNR asymptotic analysis cannot be applied to the ergodic data rate of node $n$ under NISIC, since the presence of leftover interference introduces a random term whose reciprocal moments are infinite in the considered fading model. It renders the Jensen-type upper bound divergent and the high SNR series expansion non-convergent. The ergodic data rate converges to a finite saturation value, so that its high SNR slope is zero.

As known, the ergodic data rate of node $n$ under ISIC cannot be expressed in a compact closed form. However, we can obtain its upper bound using Jensen’s inequality, as given by 
\begin{align}\label{asmy ER of user n with pSIC jensen inequality}
\tilde R_{n,\rm{upp}}^{\rm{ISIC}} = \mathbb{E} \left[ \log\left( 1+\gamma_n \right) \right] \leq \log\left[ 1+ \mathbb{E}\left( \gamma_n \right) \right].
\end{align}

\begin{corollary}\label{corollary asym ER of user n with pSIC}
According to the above inequality, on the condition of $\rho \to \infty$, the approximate ergodic data rate expression of node $n$ under ISIC over LoS propagation links is given by
\begin{align}\label{asym ER of user n with pSIC}
\tilde R_{n,\rm{upp}}^{\rm{ISIC,\infty}} = \log \left[ {1 + \frac{{\varepsilon \eta \rho {a_n}}}{{R_n^2}}\ln \left( {\frac{{R_n^2 + {d^2}}}{{{d^2}}}} \right)} \right].
\end{align}
\end{corollary}
\begin{remark}\label{remark asym ER of user n with pSIC}
Upon plugging \eqref{asym ER of user n with pSIC} into \eqref{Definition high SNR Slope}, the high SNR slope of node $n$ under ISIC over LoS propagation links equals one.
\end{remark}

\begin{corollary}\label{corollary asym ER of user f}
In the case of $\rho \to \infty$, we can obtain the asymptotic expression for ergodic data rate of node $f$ over LoS propagation links as follows
\begin{align}\label{asym ER of user f}
\tilde R_{f,\infty}^{\rm LoS} = {\log _2}\left( {1 + \frac{{{a_f}}}{{{a_n}}}} \right).
\end{align}
\end{corollary}
\begin{remark}\label{remark asym ER of user f}
Upon plugging \eqref{asym ER of user f} into \eqref{Definition high SNR Slope}, the high SNR slope of node $f$ is zero. This implies that the ergodic data rate asymptotically approaches throughput ceiling at high SNRs, since the remaining interference from non-orthogonal nodes remains non-vanishing and offsets the growth of desired signal power.
\end{remark}
\begin{corollary}\label{corollary asym ER of user f with NLoS}
In the case of $\rho \to \infty$, the asymptotic ergodic data rate expression of node $f$ over NLoS propagation links can be presented as
\begin{align}\label{asym ER of user f with NLoS}
\tilde R_{f,\infty}^{{\rm{NLoS}} } = \frac{\pi }{{M\ln 2}}\sum\limits_{m = 1}^M {\frac{{{a_f}\sqrt {1 - t_m^2} }}{{2{a_n} + {a_f}\left( {{t_m} + 1} \right)}}}.
\end{align}
\end{corollary}
\begin{remark}\label{remark asym ER of user f with NLoS}
Upon plugging \eqref{asym ER of user f with NLoS} into \eqref{Definition high SNR Slope}, the high SNR slope of node $f$ under NLoS conditions is also zero. This is because both the desired signal and the interference from node $n$ scale with transmit power, causing the SINR to converge to a constant.
\end{remark}

\subsection{Latency-Tolerant Throughput Analyses}\label{delay-tolerant mode System throughput}
In latency-tolerant transmission model, the achievable data rate is described in terms of ergodic rate, where the propagation links of node $n$ and node $f$ can be considered to be independent of each other. At this point, the latency-tolerant system throughput of PASS-NOMA networks is defined as 
\begin{align}\label{Definition System Throughput 2}
T_{{dt}} = \tilde R_n^{\nu_1} + \tilde R_f^{\nu_2},
\end{align}
where $\tilde R_n^{\rm{NISIC}}$ and $\tilde R_n^{\rm{ISIC}}$ is obtained from \eqref{ER of user n with ipSIC} and \eqref{ER of user n with pSIC} respectively. $\tilde R_f^{LoS}$ and $\tilde R_f^{NLoS}$ can be retrieved from \eqref{ER of user f} and \eqref{ER of user f with NLoS}, separately.

\begin{table}
\centering
\caption{Key Parameters Set for PASS-NOMA Networks.}
\tabcolsep5pt
\renewcommand\arraystretch{1.2} 
\begin{tabular}{|l|l|}
\hline
\multirow{2}{*}{Power allocation factors for node $n$ and node $f$} & \multirow{1}{*}{$a_n = 0.3$} \\ & \multirow{1}{*}{$a_f = 0.7$} \\
\hline
\multirow{2}{*}{Target data rates for node $n$ and node $f$} & \multirow{1}{*}{${\hat{R}}_n = 1$ BPCU} \\ & \multirow{1}{*}{${\hat{R}}_f = 1$ BPCU} \\
\hline
\multirow{1}{*}{Carrier frequency} & \multirow{1}{*}{$f_c = 28$ GHz} \\
\hline
\multirow{1}{*}{Bandwidth} & \multirow{1}{*}{$B = 100$ MHz} \\
\hline
\multirow{1}{*}{Path loss factor} & \multirow{1}{*}{$\alpha = 2$} \\
\hline
\multirow{1}{*}{The number of pinching antennas} & \multirow{1}{*}{$K = 10$} \\
\hline
\multirow{1}{*}{The height of PASS} & \multirow{1}{*}{$d = 5$ m} \\
\hline
\multirow{1}{*}{The max radius of communication region} & \multirow{1}{*}{$R_\mathcal{D} = 10$ m} \\
\hline
\multirow{2}{*}{The distribution radius of node $n$ and $f$} & \multirow{1}{*}{$R_n = 0.6R_\mathcal{D}$} \\& \multirow{1}{*}{$R_f = R_\mathcal{D}$} \\
\hline
\multirow{1}{*}{The order of Gauss-Chebyshev quadrature} & \multirow{1}{*}{$M = 300$} \\
\hline
\end{tabular}
\label{parameter}
\end{table}
\section{Simulation Results}\label{Numerical Results}
This section provides simulation results to verify the analytical results derived above, and evaluates the effect of main system parameters, i.e., communication radius, power allocation, and height of PASS. The parameters used on the simulation are listed in Table~\ref{parameter}, with the noise power is set as $\sigma^2 = -140 + 10\log(B)$ dB and BPCU denotes bits per channel use \cite{Liu2025PASSTutorial,Hou2025PerformancePASS}. The PASS-OMA networks with LoS/NLoS propagation links are chosen as the benchmark to enable performance comparison. For a fair comparison, the PASS-OMA serves two nodes under the same total bandwidth, time, transmit power, and user set as PASS-NOMA, with each node served in one of two equal time slots at the full power budget. The rate target of each time slot is set to match the per-frame target rate of PASS-NOMA, so that the per-node OMA rate is half of the rate achieved in its own time slot. In the following simulations, the equivalent power coefficient is set to $\varepsilon = 0.72K$, which is obtained by matching the average aggregate channel gain of $K$ pinching antennas to the equivalent channel gain for the considered geometry.
\subsection{Outage Probability}

\begin{figure}[t!]
\centering
 \includegraphics[width= 3.2in, height=2.4in]{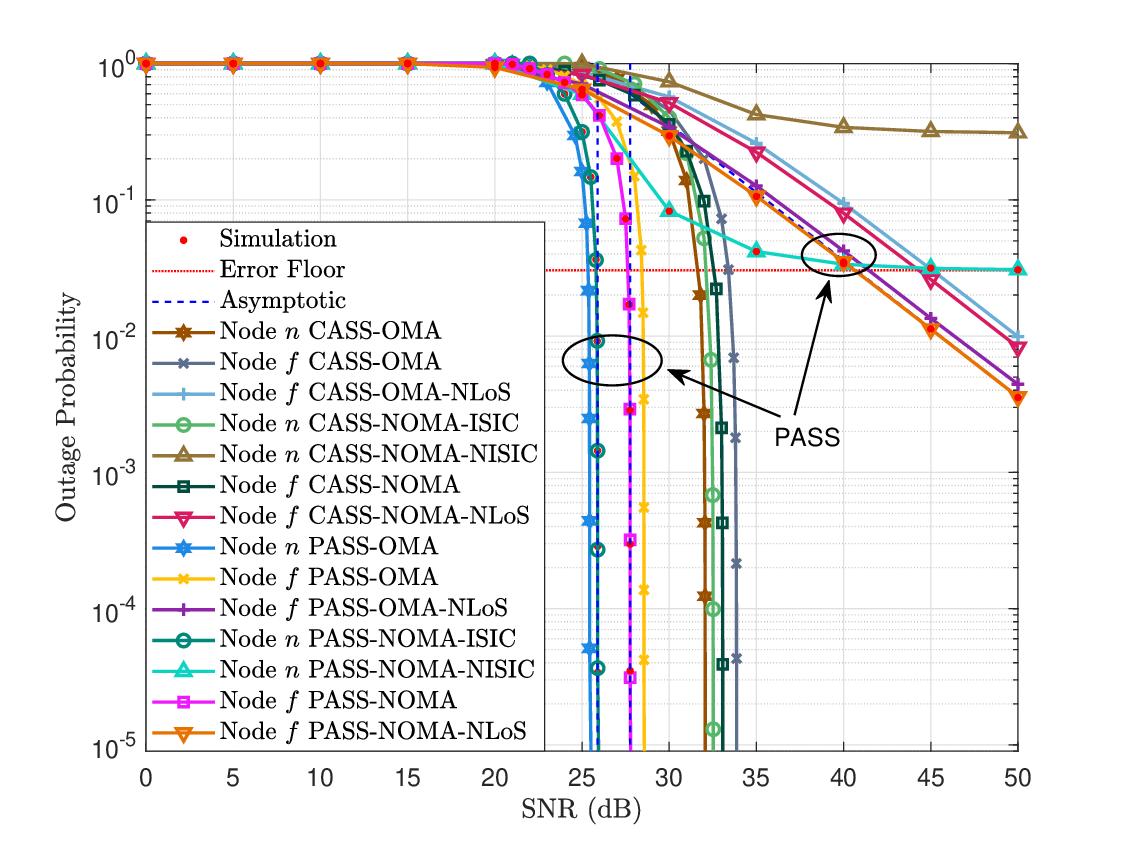}
 \caption{The outage probability versus $\rho$, with $R_{\mathcal{D}} = 10$~m, $d = 5$~m, $K = 10$, and ${\hat{R}}_n = {\hat{R}}_f = 1$ BPCU.}
\label{Fig 2}
\end{figure}

Fig.~\ref{Fig 2} plots the outage probability versus $\rho$ in PASS-NOMA networks with $R_{\mathcal{D}} = 10$ m, $K = 10$, and ${\hat{R}}_n = {\hat{R}}_f = 1$ BPCU.
The cyan and teal solid curves representing the outage probability of node $n$ under NISIC and ISIC are generated using \eqref{OP of user n with ipSIC} and \eqref{OP of user n with pSIC}, respectively.
The magenta and orange solid curves for outage probability of node $f$ over LoS and NLoS propagation links are obtained from \eqref{OP of user f} and \eqref{OP of user f with NLoS}, respectively. It can be seen that the Monte Carlo simulation results of PASS-NOMA networks match the analytical derivations. The blue dashed lines for asymptotic outage probability of node $n$ with ISIC and node $f$ with LoS/NLoS propagation links converge to the theoretical analyses, which can be obtained from \eqref{asym OP of user n with pSCI}, \eqref{asym OP of user f with LoS} and \eqref{asym OP of user f with NLoS}, respectively. 
As shown in the figure, the outage performance of node $n$ and node $f$ in PASS-NOMA networks outperforms that of PASS-OMA networks. Similarly, the outage performance of CASS-NOMA networks is mainly superior to that of CASS-OMA networks. Take it a step further, the outage performance of PASS networks is superior to that of CASS networks nearly on the whole SNR region, because of PASS design characteristics and superior architecture. Due to leftover interference, the asymptotic outage probability of node $n$ under NISIC approaches an error floor at high SNRs, which becomes much larger with increasing of interference value. In addition, PASS's LoS propagation links significantly improves outage behaviors, reflecting the path loss reduction and more favorable channel dynamics achieved by adaptive pinching antenna activation. These observations confirm that the superiority of PASS-NOMA networks highlights the potential of the PASS technique for reliable transmission in spatially random networks.

\begin{figure}[t!]
\centering
 \includegraphics[width= 3.2in, height=2.4in]{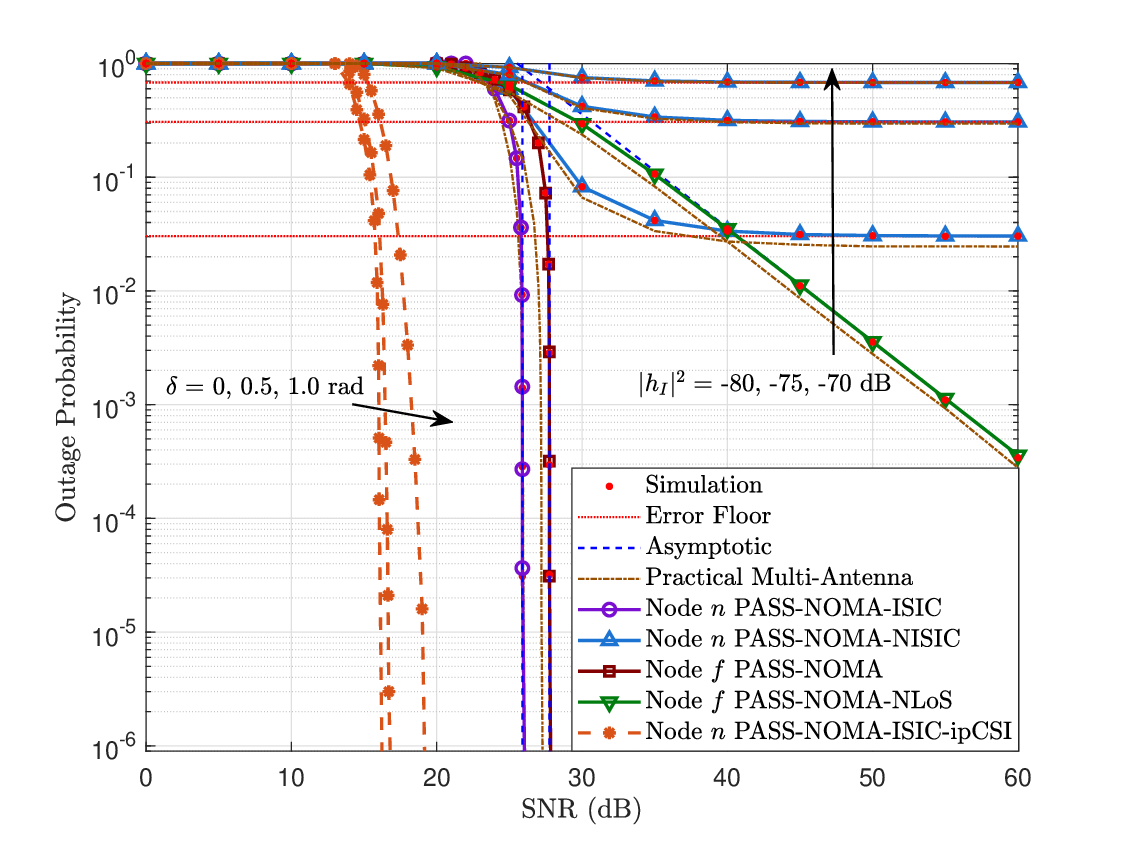}
 \caption{The outage probability versus $\rho$ under different residual interference levels and residual phase errors, with $R_{\mathcal{D}} = 10$~m, $d = 5$~m, $K = 10$, and ${\hat{R}}_n = {\hat{R}}_f = 1$ BPCU.}
\label{Fig 3}
\end{figure}

Fig.~\ref{Fig 3} plots the outage probability versus $\rho$ under different residual interference and phase-error levels, with $R_{\mathcal{D}} = 10$ m, $K = 10$, and ${\hat{R}}_n = {\hat{R}}_f = 1$ BPCU. The three groups of solid curves for node $n$ under NISIC correspond to $|h_I|^2 = -80$, $-75$, and $-70$ dB. One can observe that outage probability of node $n$ with NISIC coverage to an error floor at high SNR, while the outage performance declines with the residual interference increasing. The star-marked curves show the impact of a per-antenna residual phase error $\theta_k \sim \mathcal{U}(-\delta,\delta)$ for $\delta = 0$, $0.5$, and $1$ rad. It can be observed that the curves lie far to the left of the equivalent model benchmark, which indicates a considerable coherent-combining potential among the PASS, eroded only slightly at the considered error levels. In the limit of fully random phases, this gain vanishes and the results reduce to the incoherent power aggregation retained by the equivalent model. More important, dash-dotted curves, i.e., the practical multi-antenna results obtained from the aggregate channel $\sum_{k=1}^{K}|h_{\varphi,k}|^2$, agree well with the analytical curves and thereby validate the equivalent model as a reliable and robust characterization of PASS-NOMA networks. The additional potentiality offered by the phase alignment deserves further exploitation in future deployments of PASS.

\begin{figure}[t!]
\centering
 \includegraphics[width= 3.2in, height=2.4in]{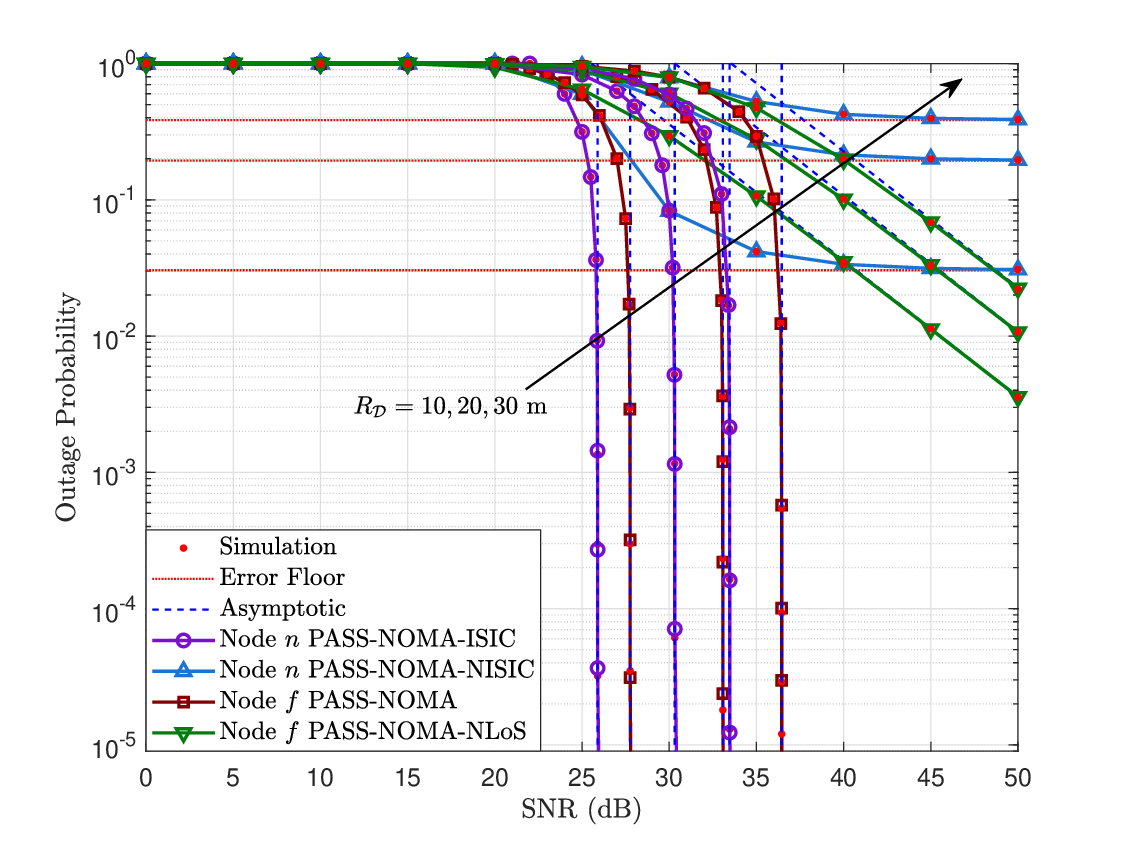}
 \caption{The outage probability versus $\rho$, with $R_{\mathcal{D}} = 10,~20,~30$~m, $d = 5$~m, $K = 10$, and ${\hat{R}}_n = {\hat{R}}_f = 1$ BPCU.}
\label{Fig 4}
\end{figure}

\begin{figure}[t!]
\centering
 \includegraphics[width= 3.2in, height=2.4in]{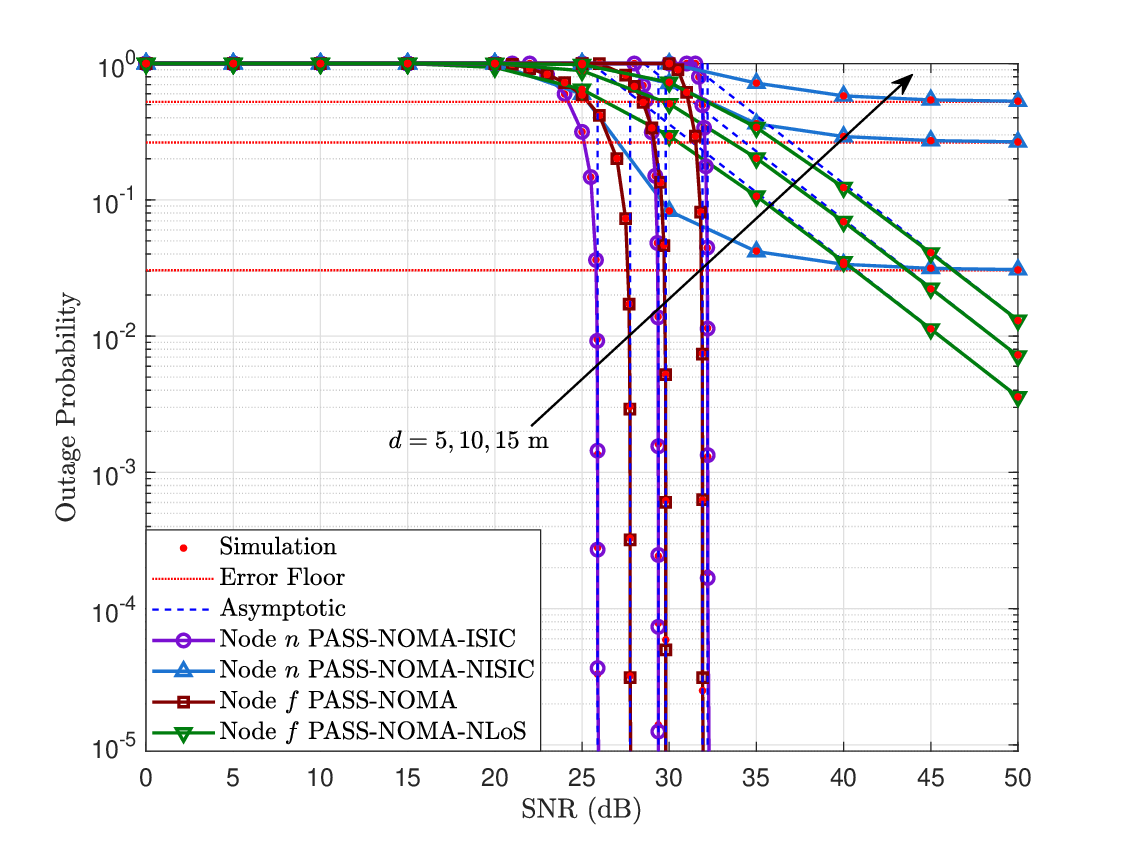}
 \caption{The outage probability versus $\rho$, with $R_{\mathcal{D}} = 10$~m, $d = 5,~10,~15$~m, $K = 10$, and ${\hat{R}}_n = {\hat{R}}_f = 1$ BPCU.}
\label{Fig 5}
\end{figure}

Fig.~\ref{Fig 4} plots the outage probability of node $n$ and node $f$ versus $\rho$ for different communication region radii, with $K = 10$ and ${\hat{R}}_n = {\hat{R}}_f = 1$ BPCU. It can be observed from the figure that as the communication radius expands, i.e., $R_{\mathcal{D}} = 10$, $20$, and $30$ m, the outage probability of PASS-NOMA networks is getting much larger. This phenomenon indicates that while a larger communication radius enhances the flexibility of coverage, more sophisticated pinching antenna activation scheme or power control strategies are needed to maintain reliable performance. Fig.~\ref{Fig 5} further plots the outage probability versus $\rho$ with different heights of the PASS, i.e., $d = 5$, $10$, and $15$ m. We can observe that the propagation distance of both nodes enlarges and the average channel gain decreases as $d$ increases, so that the outage probability of PASS-NOMA networks increases accordingly. These observations indicate that the communication radius and PASS' height jointly determine the outage behavior, which should be taken into consideration in the deployment and performance optimization.

\begin{figure}[t!]
  \centering
  \begin{subfigure}[b]{0.40\textwidth}
    \centering
    \includegraphics[width=\textwidth,height=0.6\textwidth]{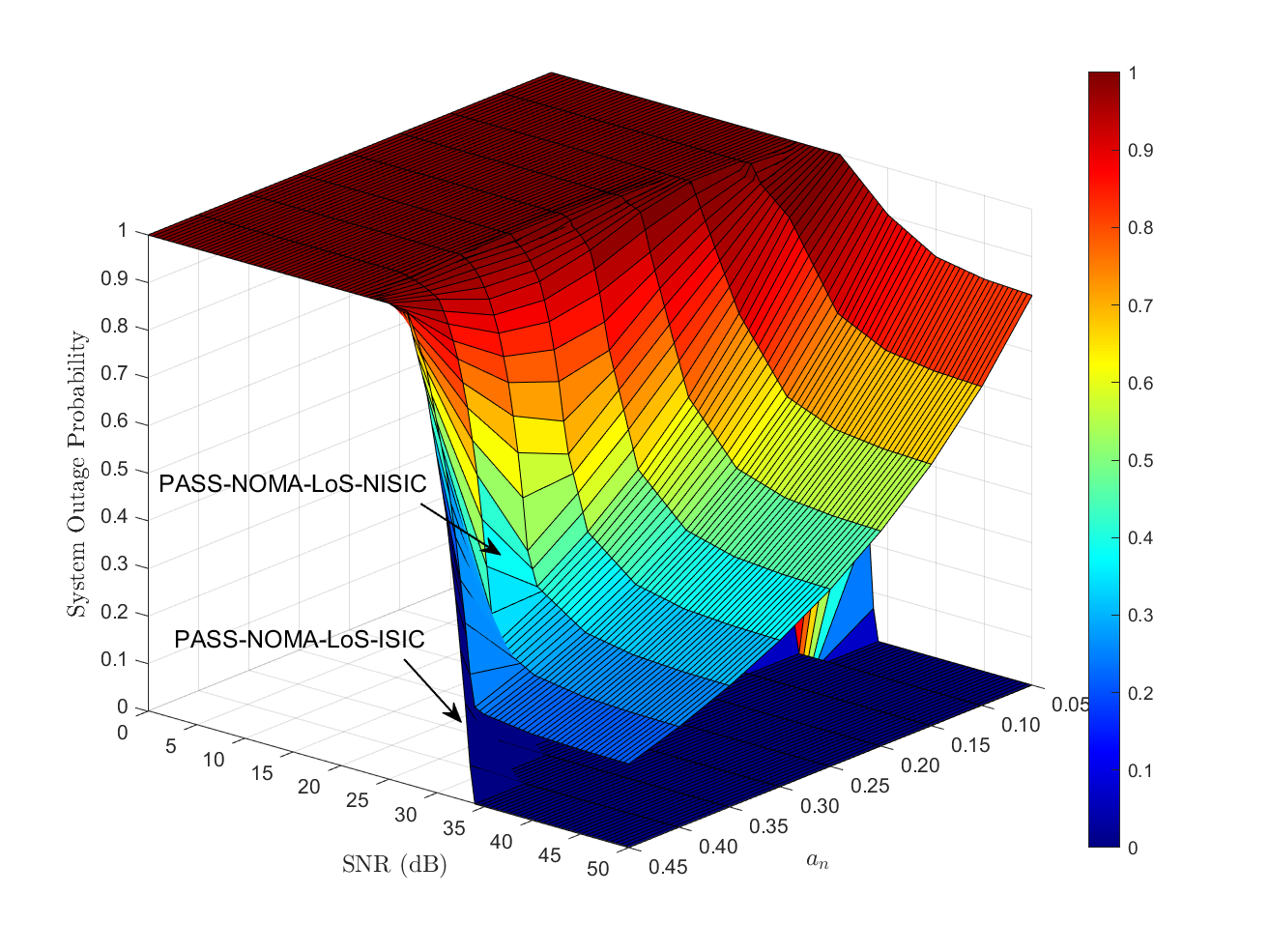}
    \caption{The PASS-NOMA-LoS networks.}
    \label{Fig 6_1}
  \end{subfigure}\hfill%
  \begin{subfigure}[b]{0.40\textwidth}
    \centering
    \includegraphics[width=\textwidth,height=0.6\textwidth]{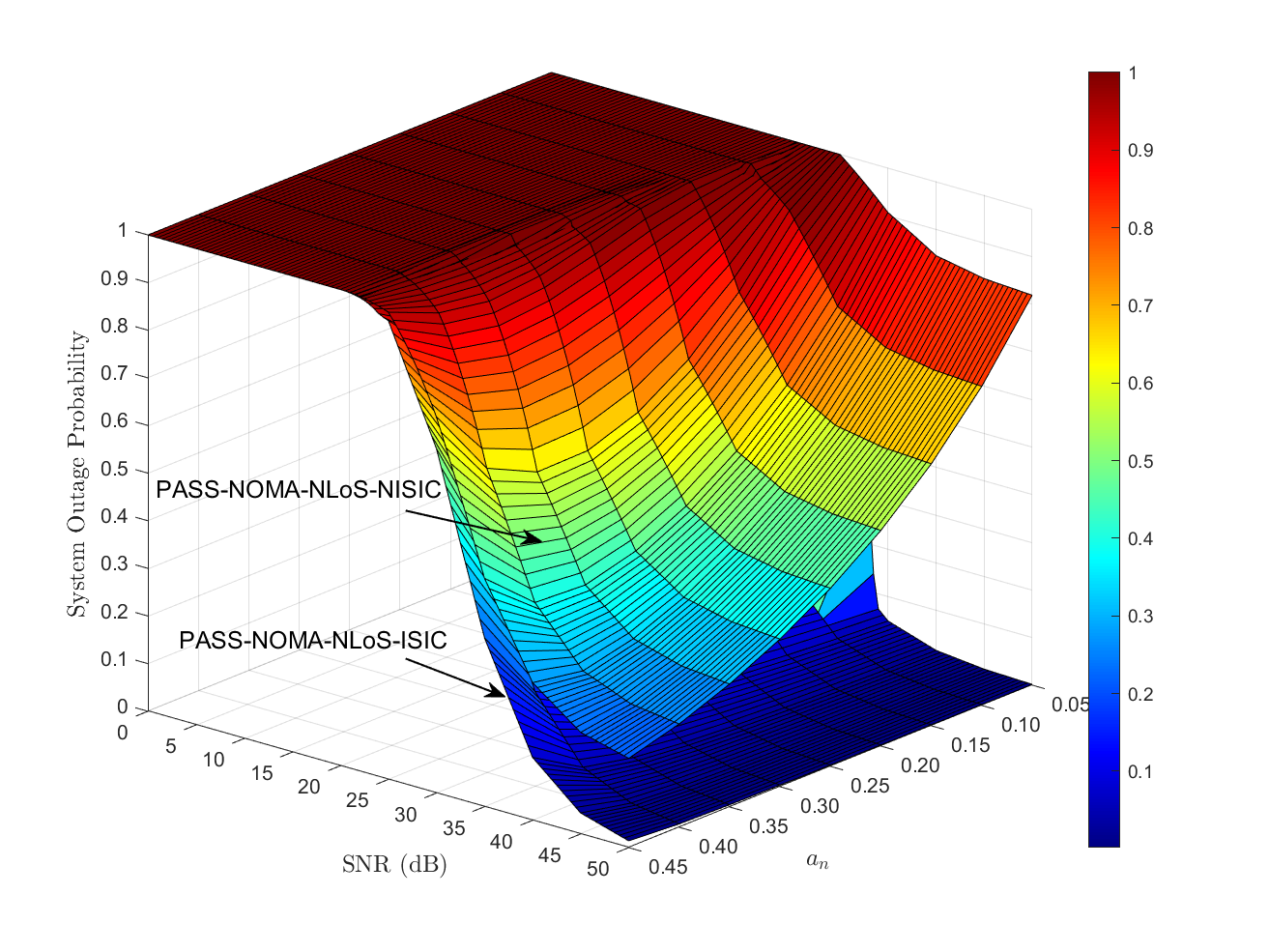}
    \caption{The PASS-NOMA-NLoS networks.}
    \label{Fig 6_2}
  \end{subfigure}
  \caption{The system outage probability versus the SNR and the power allocation coefficient $a_n$ with $R_{\mathcal{D}} = 10$~m, $d = 5$~m, $K = 10$, and ${\hat{R}}_n = {\hat{R}}_f = 1$ BPCU.}
  \label{Fig_6}
\end{figure}

Fig.~\ref{Fig_6} presents the system outage probability versus the SNR and the power allocation coefficient $a_n$ over LoS/NLoS propagation links, where the two surfaces in each panel correspond to ISIC and NISIC with $|h_I|^2 = -75$ dB, respectively. It can be observed from Fig.~\ref{Fig_6} (a) that the surfaces decrease with the SNR, and the NISIC surfaces flatten at high SNRs due to the residual interference. With the increase of $a_n$, the outage probability of node $n$ decreases since a larger power allocation coefficient improves its SIC reliability, while the power allocated to node $f$ shrinks so that its outage probability grows. For this reason, the system outage probability exhibits a trade-off valley along $a_n$ within the feasible region. In Fig.~\ref{Fig_6} (b), the system outage performance of PASS-NOMA over NLoS is worse than that of PASS-NOMA over LoS, which is consistent with the outage performance of PASS-NOMA reported above. These observations provide useful guidelines for the power allocation design of PASS-NOMA networks and corroborate the reasonableness of the adopted power split, which is essential for balancing the SIC reliability of near node and the service quality of far node, thereby optimizing the overall system reliability.

\subsection{Ergodic Data Rate}

\begin{figure}[t!]
\centering
 \includegraphics[width= 3.2in, height=2.4in]{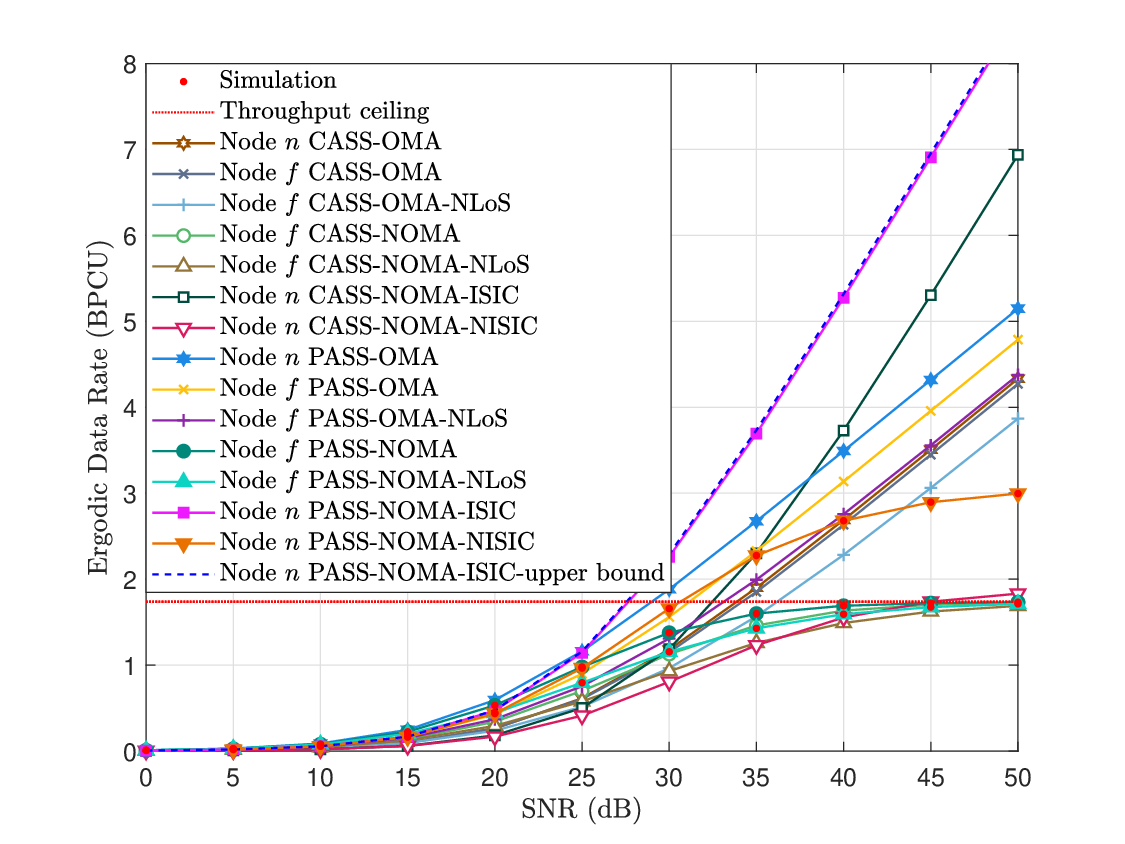}
 \caption{The ergodic data rate versus $\rho$, with $R_{\mathcal{D}} = 10$~m, $d = 5$~m, $K = 10$, and ${\hat{R}}_n = {\hat{R}}_f = 1$ BPCU.}
\label{Fig 7}
\end{figure}

Fig.~\ref{Fig 7} plots the ergodic data rates versus $\rho$, with $R_{\mathcal{D}} = 10$ m, $K = 10$, and ${\hat{R}}_n = {\hat{R}}_f = 1$ BPCU. The magenta square-marked and orange down-triangle-marked solid curves representing the ergodic data rates of node $n$ under ISIC and NISIC are generated using \eqref{ER of user n with pSIC} and \eqref{ER of user n with ipSIC}, respectively. The teal circle-marked and cyan up-triangle-marked solid curves for the ergodic data rates of node $f$ over LoS and NLoS propagation links are obtained from \eqref{ER of user f} and \eqref{ER of user f with NLoS}, respectively. The simulation curves completely coincide with the theoretical analysis results. The asymptotic curves of the ergodic data rates for node $n$ under ISIC and node $f$ over LoS/NLoS propagation links are plotted in accordance with \eqref{asym ER of user n with pSIC}, \eqref{asym ER of user f} and \eqref{asym ER of user f with NLoS}, separately. The asymptotic curves match the theoretical curves in the high SNR region. 
It can be observed that the ergodic performance of PASS-NOMA outperforms that of PASS-OMA and CASS-OMA/NOMA at high SNRs. Besides, the ergodic data rate of node $n$ with NISIC and node $f$ over LoS/NLoS links converge to a throughput ceiling. This is because that the residual interference would increase with transmit power and the ergodic data rate of node $f$ approaches throughput ceiling $\log_2(1+a_f/a_n)$ at high SNR region. These results highlight the spectral efficiency potential of PASS technique in future wireless networks.

\begin{figure}[t!]
\centering
 \includegraphics[width= 3.2in, height=2.4in]{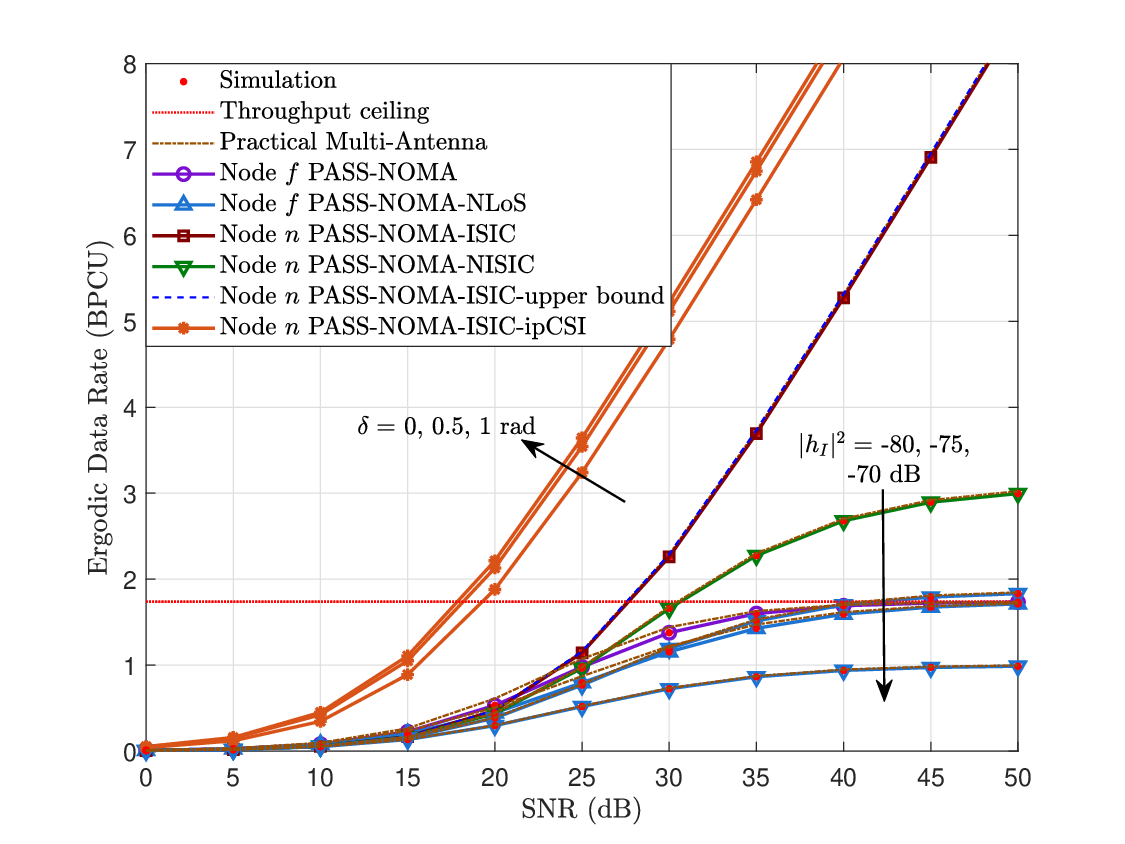}
 \caption{The ergodic data rate versus $\rho$ under different residual interference levels and residual phase errors, with $R_{\mathcal{D}} = 10$~m, $d = 5$~m, $K = 10$, and ${\hat{R}}_n = {\hat{R}}_f = 1$ BPCU.}
\label{Fig 8}
\end{figure}

Fig.~\ref{Fig 8} plots the ergodic data rate versus $\rho$ under different residual phase errors, together with the practical multi-antenna realization, with $R_{\mathcal{D}} = 10$ m, $K = 10$, and ${\hat{R}}_n = {\hat{R}}_f = 1$ BPCU. It can be observed that the analytical results of ergodic data rates closely reproduce the dash-dotted curves obtained from the aggregate channel $\sum_{k=1}^{K}|h_{\varphi,k}|^2$, which validates the accuracy of equivalent channel models. Additionally, the star-marked curves with $\delta = 0$, $0.5$, and $1$ rad is shown in figure, which reveal the impact of the residual phase errors on the ergodic performance. The case of $\delta = 0$ attains the highest ergodic data rate, and the residual phase errors erode part of the coherent combining gain as $\delta$ increases, which is relevant to coherent channel and different with incoherent channel in our system model. In addition, we can observe that the ergodic data rate of node $n$ under NISIC gradually saturates due to the residual interference, and that of node $n$ approach the throughput ceiling. These results reveal that the ergodic performance of PASS is robust to the residual phase errors to a certain extent, and maintaining a proper level of phase alignment is a promising direction for further improving the spectral efficiency of PASS-NOMA networks.

\begin{figure}[h!]
\centering
 \includegraphics[width= 3.2in, height=2.4in]{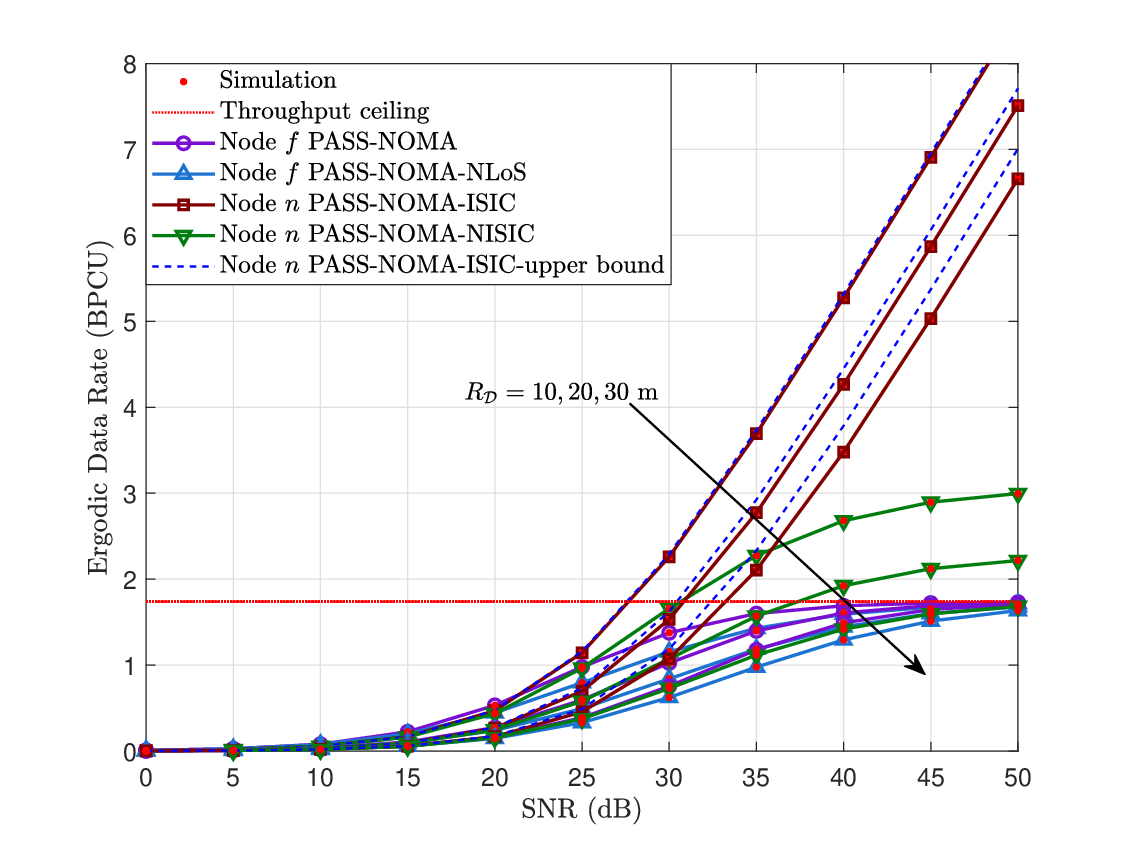}
 \caption{The ergodic data rate versus $\rho$, with $R_{\mathcal{D}} = 10,~20,~30$~m, $d = 5$~m, $K = 10$, and ${\hat{R}}_n = {\hat{R}}_f = 1$ BPCU.}
\label{Fig 9}
\end{figure}

\begin{figure}[h!]
\centering
 \includegraphics[width= 3.2in, height=2.4in]{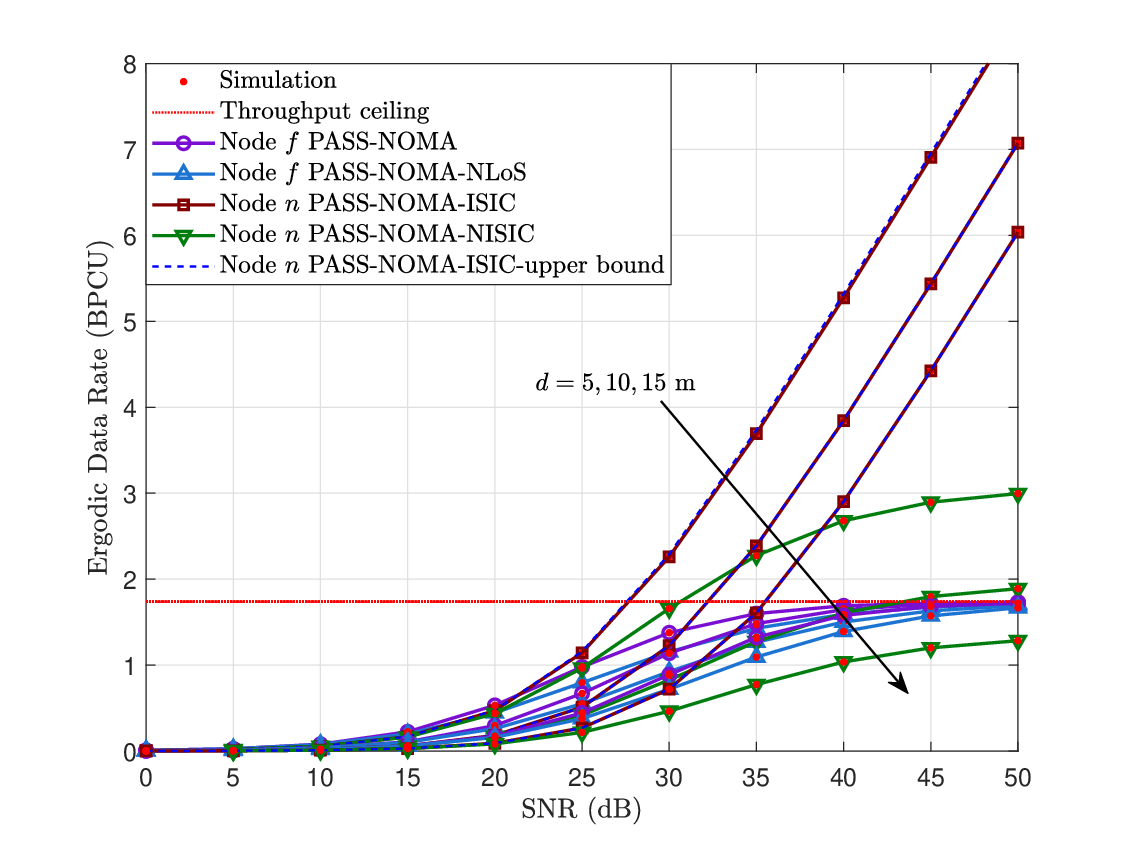}
 \caption{The ergodic data rate versus $\rho$, with $R_{\mathcal{D}} = 10$~m, $d = 5,~10,~15$~m, $K = 10$, and ${\hat{R}}_n = {\hat{R}}_f = 1$ BPCU.}
\label{Fig 10}
\end{figure}

Fig.~\ref{Fig 9} plots the ergodic data rate of node $n$ and node $f$ versus $\rho$ for different communication region radii, with $K = 10$, ${\hat{R}}_n = {\hat{R}}_f = 1$ BPCU. It can be observed that as $R_{\mathcal{D}}$ increases, the ergodic data rates of PASS-NOMA networks decrease substantially due to the increase in path loss and the weakened average received SNR. We also observe that the ergodic performance of node $f$ exhibits more noticeable degradation. The reason for this phenomenon is that node $f$ is susceptible to path occlusion and propagation loss. These simulation results indicate that it is important to highlight the coupling relationship between communication coverage and the ergodic data rates of PASS-NOMA networks. 
Fig.~\ref{Fig 10} further presents the ergodic data rate of node $n$ and node $f$ versus $\rho$ for different heights of the PASS, i.e., $d = 5$, $10$, and $15$ m. Similar to the effect of the communication radius, the ergodic data rates of both nodes decrease as $d$ increases, since the enlarged propagation distance weakens the average received SNR. These observations further indicate that communication coverage and deployment height should be jointly considered in ergodic performance optimization of PASS-NOMA networks.

\subsection{System Throughput}

\begin{figure}[t!]
\centering
 \includegraphics[width= 3.2in, height=2.4in]{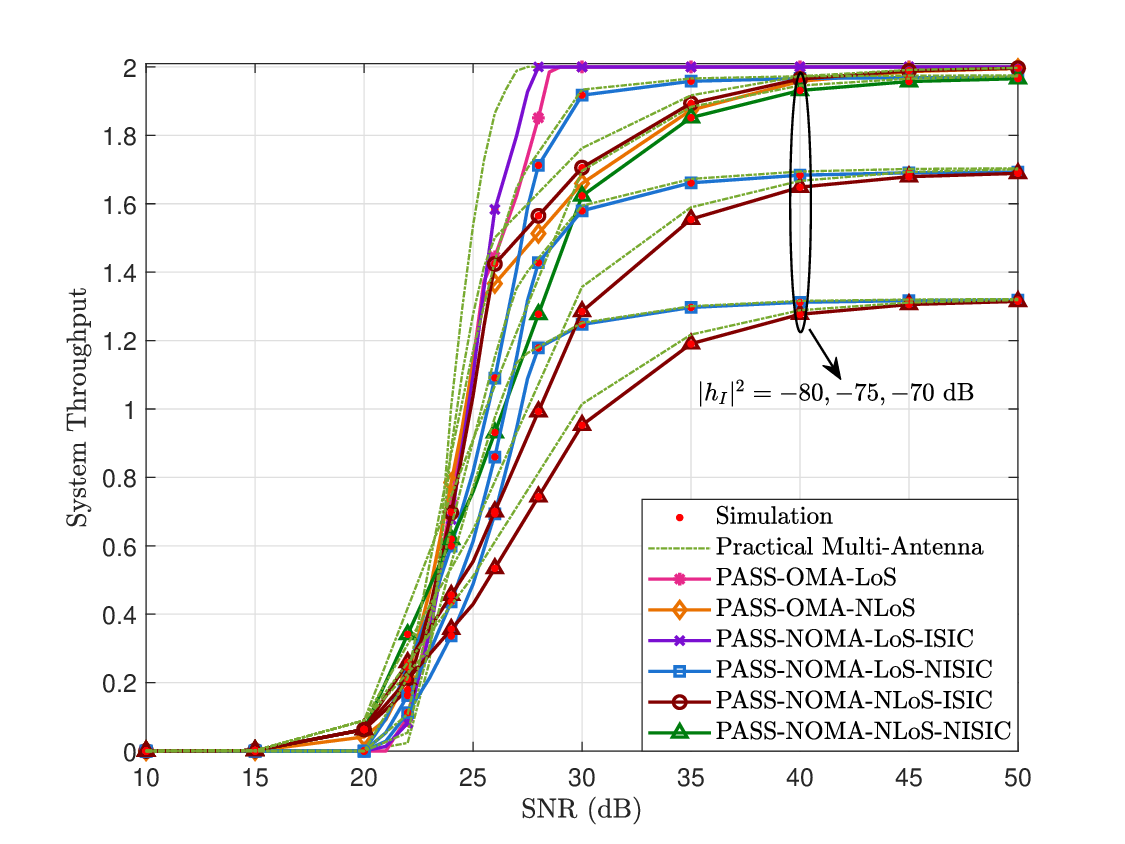}
 \caption{The system delay-constrained throughput versus $\rho$, with $R_{\mathcal{D}} = 10$~m, $d = 5$~m, $K = 10$, and ${\hat{R}}_n = {\hat{R}}_f = 1$ BPCU.}
\label{Fig 11}
\end{figure}

Fig.~\ref{Fig 11} plots the delay-constrained system throughput versus $\rho$, with $K = 10$, and ${\hat{R}}_n = {\hat{R}}_f = 1$ BPCU. The delay-constrained system throughput curves of PASS-NOMA networks are plotted based \eqref{Definition System Throughput}. One can observe that the delay-constrained system throughput of PASS-NOMA networks exceeds that of PASS-OMA before achieving the maximum. Another important phenomenon is that PASS-NOMA networks reaches the frame ceiling at a lower SNR than PASS-OMA. Under NISIC scheme, the system throughput is lower than that under ISIC scheme in high SNR region where the residual interference induces an error floor. The throughput ceiling decreases as the residual interference becomes stronger for $|h_I|^2 = -80$, $-75$, and $-70$ dB, respectively. Moreover, it is obtained well agreement with the analytical curves comparison to the practical multi-antenna results. In addition, Fig.~\ref{Fig 12} plots the latency-tolerant system throughput versus $\rho$, with $K = 10$, and ${\hat{R}}_n = {\hat{R}}_f = 1$ BPCU. The latency-tolerant system throughput curves of PASS-NOMA networks are plotted based \eqref{Definition System Throughput 2}. It can be observed that latency-tolerant system throughput increases monotonically with $\rho$ under ISIC, since the ergodic data rates of both nodes keep growing. On the contrary, the overall throughput flattens, which is attributed to the fact that the rate of node $n$ under NISIC gradually saturates due to the residual interference. The throughput ceiling will decreases as the residual interference becomes stronger. These observations highlight that although latency-tolerant transmission relaxes the outage constraints, the overall system throughput remains fundamentally limited by the multiplexing structure and the residual interference, which underlines the importance of enhancing the SIC quality and interference suppression for exploiting the throughput potential of PASS-NOMA networks.

\begin{figure}[t!]
\centering
 \includegraphics[width= 3.2in, height=2.4in]{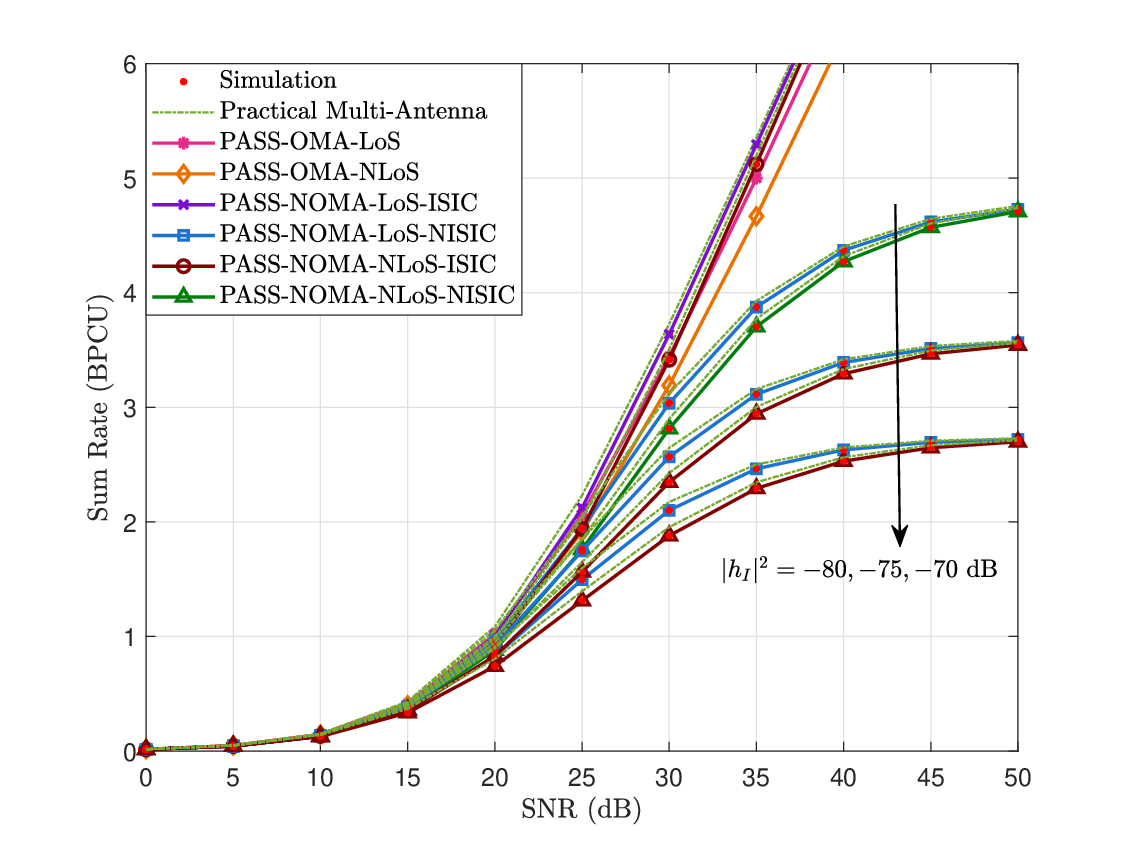}
 \caption{The system latency-tolerant throughput versus $\rho$, with $R_{\mathcal{D}} = 10$~m, $d = 5$~m, $K = 10$, and ${\hat{R}}_n = {\hat{R}}_f = 1$ BPCU.}
\label{Fig 12}
\end{figure}

\section{Conclusion}\label{Conclusion}
This paper has examined the performance of PASS-NOMA networks over LoS/NLoS propagation links with randomly distributed nodes. 
The outage probability and ergodic data rate with ISIC/NISIC in the closed-form and approximate expressions were derived in PASS-NOMA networks. 
Based on these analytical results, the diversity gains and slopes of a couple of nodes over LoS/NLoS propagation links have been attained.
Both delay-constrained and latency-tolerant system throughput were further discussed to capture performance under different latency requirements. 
Numerical results have validated the reliability of the analytical derivations and demonstrated that PASS-NOMA networks outperforms PASS-OMA and CASS-OMA/NOMA in system reliability. It verified that the analytical performance of our equivalent model closely reproduce that the practical multi-antenna aggregate channel model.

\appendices
\section*{Appendix~A} \label{Appendix:A}
\renewcommand{\theequation}{A.\arabic{equation}}
\setcounter{equation}{0}
Upon plugging \eqref{SINR n to f} and \eqref{SINR n} into \eqref{OP of user n defined}, after a simple transfer calculation and let $C_n^{\mathrm{NISIC}}(x) = \frac{{\varepsilon \eta \rho {a_n}}}{{{\gamma _{thn}}\left( {\rho x + 1} \right)}}$, and $X = {{\left| {{h_I}} \right|}^2}$, the outage probability expression of node $n$ with NISIC can be written as 
\begin{align}
{\rm{P}}_n^{\rm{NISIC}}(\rho ) &= {\rm{P}}\left[ {{x_n^2} + y_n^2 > C_n^{\rm{NISIC}}(x) - {d^2}} \right] \nonumber \\
&= 1-F_{r^2_n}\left[ C_n^{\rm{NISIC}}(x) - {d^2} \right],
\end{align}
where $F_{r^2_n}(\cdot)$ denotes the CDF of the square of the distance from node $n$ to origin $\left(0,0\right)$. Since the PDF of node's location, i.e., $f_{{\mathbf{u}}_{\varphi}} (\cdot) = \frac{1}{\pi R_{\varphi}^2}$,  the corresponding CDF of distance from node $\varphi$ to origin can be calculated as 
\begin{align}
F_{r_{\varphi}}(r) = \int_{0}^{2 \pi} {\int_{0}^{r} {\frac{1}{\pi R_{\varphi}^2} \rho \rm{d\rho} }} 
=\frac{r^2}{R_{\varphi}^2},~{0 \le r \le {R_{\varphi}^2}}.
\end{align}
After that, by employing $F_Y(y)=F_Z(\sqrt{y})$, where $Y=r^2_{\varphi}$, and $Z=r_{\varphi} > 0$, the CDF of square of distance from node $\varphi$ to origin can be finally given by
\begin{align}\label{The CDF of quare of distance}
F_{r^2_{\varphi}}(y) = \begin{cases}
0, & {y<0}, \\
\frac{y}{R^2_{\varphi}}, & {0 \le y < {R^2_{\varphi}}}, \\
1, & {y \ge {R^2_{\varphi}}}.
\end{cases}
\end{align}

Upon plugging $C_n^{\mathrm{NISIC}}(x)$, and \eqref{The CDF of quare of distance} into the outage probability expression of node $n$ with NISIC, we can achieve the elementary derivation formula as 
\begin{align}\label{Appendx A A_4}
&{\rm{P}}_n^{{\rm{NISIC}}}(\rho ) = \int_0^\infty  {{\mathrm{P}}_n^{{\rm{NISIC}}}\left[\rho ,X(x)\right]{f_{|{h_I}{|^2}}}(x)dx} \nonumber \\ 
&= \begin{cases}
     1, &{C_n^{\mathrm{NISIC}}\left( 0 \right)} - d^2 < 0, \\
     {J_{A1}} + {J_{A3}}, &0 < C_n^{\mathrm{NISIC}}\left( 0 \right) - d^2 < R_n^2, \\
     {J_{A2}} + {J_{A3}} , &C_n^{\mathrm{NISIC}}\left( 0 \right) - d^2 > R_n^2,
   \end{cases}
\end{align}
where 
${J_{A1}} = \int_0^{\frac{{\varepsilon \eta {a_n}}}{{{\gamma _{thn}}{d^2}}} - \frac{1}{\rho }} {(1 - \frac{{\varepsilon \eta \rho {a_n} - {\gamma _{thn}}\left( {\rho x + 1} \right){d^2}}}{{{\gamma _{thn}}\left( {\rho x + 1} \right)R_n^2}}){f_{|{h_I}{|^2}}}(x)dx} $,
${J_{A2}} = {\int_{\frac{{\varepsilon \eta {a_n}}}{{{\gamma _{thn}}(R_n^2 + {d^2})}} - \frac{1}{\rho }}^{\frac{{\varepsilon \eta {a_n}}}{{{\gamma _{thn}}{d^2}}} - \frac{1}{\rho }} {(1 - \frac{{\varepsilon \eta \rho {a_n} - {\gamma _{thn}}\left( {\rho x + 1} \right){d^2}}}{{{\gamma _{thn}}\left( {\rho x + 1} \right)R_n^2}}){f_{|{h_I}{|^2}}}(x)dx} }$,
${J_{A3}} = \int_{\frac{{\varepsilon \eta {a_n}}}{{{\gamma _{thn}}{d^2}}} - \frac{1}{\rho }}^\infty  {{f_{|{h_I}{|^2}}}(x)dx} $. And for convenience, let $C_n = C_n^{\mathrm{NISIC}}\left( 0 \right)=\frac{{\varepsilon \eta \rho {a_n}}}{{{\gamma _{thn}}}}$. Since $h_I$ is modeled as a Rayleigh random variable, its modulus squared CDF is ${F_{|{h_I}{|^2}}}\left( x \right) = 1 - e^{-\frac{x}{\Omega_I}}$. The formula ${J_{A3}}$ is easily obtained as ${J_{A3}} = 1 - {F_{|{h_I}{|^2}}}\left[ {\frac{{\varepsilon \eta {a_n}}}{{{\gamma _{thn}}{d^2}}} - \frac{1}{\rho }} \right] = {e^{ - \frac{1}{{{\Omega _I}}}\left( {\frac{{\varepsilon \eta {a_n}}}{{{\gamma _{thn}}{d^2}}} - \frac{1}{\rho }} \right)}}$. By conducting some basic mathematical calculations, formula ${J_{A1}}$ can be written as
\begin{align}\label{Appendx A A_5}
{J_{A1}} =& \left( {1 + \frac{{{d^2}}}{{R_n^2}}} \right) {{F_{|{h_I}{|^2}}}\left( {\frac{{\varepsilon \eta {a_n}}}{{{\gamma _{thn}}{d^2}}} - \frac{1}{\rho }} \right)} \nonumber \\
& - \underbrace{\int_0^{\frac{{\varepsilon \eta {a_n}}}{{{\gamma _{\mathrm{thn}}}{d^2}}} - \frac{1}{\rho }} {\frac{{\varepsilon \eta \rho {a_n}}}{{{\gamma _{thn}}\left( {\rho x + 1} \right)R_n^2}}\frac{1}{{{\Omega _I}}}{e^{ - {\textstyle{1 \over {{\Omega _I}}}}x}}dx}}_{J_{A12}}.
\end{align}
The term ${J_{A12}}$ is a simple definite integral, which can be directly evaluated in closed form as
\begin{align}
{J_{A12}} = \frac{{\varepsilon \eta {a_n}}}{{{\Omega _I}{\gamma _{thn}}R_n^2}}{e^{\frac{1}{{{\Omega _I}\rho }}}}\left[ {{\rm{Ei}}\left( { - \frac{{{\xi _{A1}}}}{{{\Omega _I}}}} \right) - {\rm{Ei}}\left( { - \frac{1}{{{\Omega _I}\rho }}} \right)} \right],
\end{align}
where ${\xi _{A1}} = \frac{{\varepsilon \eta {a_n}}}{{{\gamma _{thn}}{d^2}}}$ and ${\rm{Ei}}\left( \cdot \right)$ is the exponential integral function\cite[Eq. (8.211.1)]{2000gradshteyn}. Similarly, the closed-form expression of ${J_{A2}}$ can be obtained by straightforward definite integration, it can be written as
\begin{align}\label{Appendx A A_7}
{J_{A2}} =& \left( {1 + \frac{{{d^2}}}{{R_n^2}}} \right){e^{\frac{1}{{{\Omega _I}\rho }}}}\left( {{e^{ - \frac{{{\xi _{A2}}}}{{{\Omega _I}}}}} - {e^{ - \frac{{{\xi _{A1}}}}{{{\Omega _I}}}}}} \right) \nonumber \\
&+ \frac{{\varepsilon \eta {a_n}}}{{{\Omega _I}{\gamma _{thn}}R_n^2}}{e^{\frac{1}{{{\Omega _I}\rho }}}}\left[ {{\rm{Ei}}\left( { - \frac{{{\xi _{A1}}}}{{{\Omega _I}}}} \right) - {\rm{Ei}}\left( { - \frac{{{\xi _{A2}}}}{{{\Omega _I}}}} \right)} \right],
\end{align}
where ${\xi _{A2}} = \frac{{\varepsilon \eta {a_n}}}{{{\gamma _{thn}}(R_n^2 + {d^2})}}$. After some basic algebraic operations with plugging $J_{A3}$, \eqref{Appendx A A_5}, and \eqref{Appendx A A_7} into \eqref{Appendx A A_4}, the outage probability expression \eqref{OP of user n with ipSIC} is finally obtained. The proof is completed.
\appendices
\section*{Appendix~B} \label{Appendix:B}
\renewcommand{\theequation}{B.\arabic{equation}}
\setcounter{equation}{0}
Upon plugging \eqref{SINR f} into \eqref{OP of user f defined}, after a simple transfer calculation, the outage probability expression of node $f$ over NLoS propagation links can be further calculated as 
\begin{align}
{{\mathrm{P}}_f}\left( {{P_b}} \right) = {\mathrm{P}}\left( {x_f^2 + y_f^2 > \rho \varepsilon \eta {{\left| {{h_f}} \right|}^2}(\frac{{{a_f}}}{{{\gamma _{thf}}}} - {a_n}) - {d^2}} \right).
\end{align}
The CDF of square of distance from node $f$ to origin can be obtained from \eqref{The CDF of quare of distance}, and the PDF of $Y = \left| h_f \right|^2$ which follows a Rayleigh distribution, is obtained as $f_{\left| h_f \right|^2}(y) = \frac{1}{\Omega_f} e^{-\frac{y}{\Omega_f}}$. The outage expression of node $f$ under NLoS channel conditions is presented as 
\begin{align}
{\rm{P}}_f^{\rm{NLoS}}\left( \rho  \right) = \int_0^\infty  {{\rm{P}}_f^{\rm{NLoS}}\left[ {\rho ,Y\left( y \right)} \right]{f_{{{\left| {{h_f}} \right|}^2}}}} \left( y \right){\rm{dy}},
\end{align}
where ${\rm{P}}_f^{\rm{NLoS}}\left[ {\rho ,Y\left( y \right)} \right] = F_{r^2_f}^{\rm{NLoS}}\left[ \rho \varepsilon \eta y (\frac{{{a_f}}}{{{\gamma _{thf}}}} - {a_n}) - {d^2} \right]$, $ y= {{\left| {{h_f}} \right|}^2}$. After that, the integral above can be further expanded to obtain as 
\begin{align}\label{Appendx B B_3}
{\rm{P}}_f^{\rm{NLoS}}\left( \rho  \right) =& \underbrace {\int_0^{\delta _1} {\frac{1}{{{\Omega _f}}}{{\rm{e}}^{ - \frac{y}{{{\Omega _f}}}}}{\rm{dy}}} }_{{{\rm{J}}_{B1}}} \nonumber \\ 
& + \underbrace {\int_{\delta _1}^{\delta _2} {\left( {1 - \frac{{\rho \varepsilon \eta y(\frac{{{a_f}}}{{{\gamma _{thf}}}} - {a_n}) - {d^2}}}{{R_f^2}}} \right)\frac{1}{{{\Omega _f}}}{{\rm{e}}^{ - \frac{y}{{{\Omega _f}}}}}{\rm{dy}}} }_{{J_{B2}}} \nonumber \\ 
& + \int_{\delta _2}^\infty  {0 \cdot {f_{{{\left| {{h_f}} \right|}^2}}}\left( y \right){\rm{dy}}},
\end{align}
where ${\delta _1} = \frac{{{d^2}}}{{\rho \varepsilon \eta ({\textstyle{{{a_f}} \over {{\gamma _{thf}}}}} - {a_n})}},{\delta _2} = \frac{{R_f^2 + {d^2}}}{{\rho \varepsilon \eta ({\textstyle{{{a_f}} \over {{\gamma _{thf}}}}} - {a_n})}}$. The formula ${J_{B1}}$ can be directly obtained as ${J_{B1}} = {F_{{{\left| {{h_f}} \right|}^2}}}\left[ {\frac{{{d^2}}}{{\rho \varepsilon \eta ({\textstyle{{{a_f}} \over {{\gamma _{thf}}}}} - {a_n})}}} \right] = 1 - {e^{ - \frac{{{\delta _1}}}{{{\Omega _f}}}}}$. By conducting some basic mathematical calculations, formula ${J_{B2}}$ can be written as 
\begin{align}\label{Appendx B B_4}
{J_{B2}} =& \left( {1 + \frac{{{d^2}}}{{R_f^2}}} \right)\left( {{e^{ - \frac{{{\delta _1}}}{{{\Omega _f}}}}} - {e^{ - \frac{{{\delta _2}}}{{{\Omega _f}}}}}} \right) \nonumber \\ 
& - \underbrace {\int_{{\delta _1}}^{{\delta _2}} {\frac{{\rho \varepsilon \eta y(\frac{{{a_f}}}{{{\gamma _{thf}}}} - {a_n})}}{{R_f^2}}\frac{1}{{{\Omega _f}}}{{\rm{e}}^{ - \frac{y}{{{\Omega _f}}}}}{\rm{dy}}} }_{{J_{B22}}}.
\end{align}
After a straightforward calculation, the formula $J_{B22}$ can be expressed as 
\begin{align}
J_{B22} &= \frac{{\rho \varepsilon \eta }}{{R_f^2}}(\frac{{{a_f}}}{{{\gamma _{thf}}}} - {a_n}) {\int_{{\delta _1}}^{{\delta _2}} {\frac{y}{{{\Omega _f}}}{{\rm{e}}^{ - \frac{y}{{{\Omega _f}}}}}{\rm{dy}}} } \nonumber \\ 
&= \frac{{\rho \varepsilon \eta }}{{R_f^2}}(\frac{{{a_f}}}{{{\gamma _{thf}}}} - {a_n})\left[ {\left( {{\delta _1} + {\Omega _f}} \right){{\rm{e}}^{ - \frac{{{\delta _1}}}{{{\Omega _f}}}}} - \left( {{\delta _2} + {\Omega _f}} \right){{\rm{e}}^{ - \frac{{{\delta _2}}}{{{\Omega _f}}}}}} \right].
\end{align}
Combining formula $J_{B1}$ and \eqref{Appendx B B_4} with $J_{B22}$ into \eqref{Appendx B B_3}, the final result \eqref{OP of user f with NLoS} can be acquired. The proof is completed.

\appendices
\section*{Appendix~C} \label{Appendix:C}
\renewcommand{\theequation}{C.\arabic{equation}}
\setcounter{equation}{0}
The proof can be started upon plugging \eqref{SINR n} into \eqref{Definition Ergodic Rate} under $\varpi = 1$ condition. The CDF $F_{\gamma_n} \left( x \right)$ is defined as 
\begin{align}
F_{\gamma_n} \left( x \right) = {\rm{P}}\left( \gamma_n < x \right).
\end{align} 
It can be found that the CDF is nearly identical to the outage probability expression of node $n$ under NISIC. The only difference is that $\gamma_{thn}$ in the formula is replaced with $x$. Accordingly, the ergodic data rate expression can be calculated as 
\begin{align}
&\tilde R_n^{\rm{NISIC}} = \frac{1}{{\ln 2}}\int_0^\infty  {\frac{{1 - {F_{\gamma _n^{\rm{NISIC}}}}\left( x \right)}}{{1 + x}}{\rm{dx}}} \nonumber \\ 
&= \frac{1}{{\ln 2}}\int_0^{\rho {\varsigma _1}} \frac{1}{{1 + x}}\left[ 1 - {e^{\frac{1}{{{\Omega _I}\rho }}}}\left\{ \left( {1 + \frac{{{d^2}}}{{R_n^2}}} \right){e^{ - \frac{{{\varsigma _1}}}{{x{\Omega _I}}}}} \right. \right. \nonumber \\ 
& \left. \left. - \frac{{{d^2}}}{{R_n^2}}{e^{ - \frac{{{\varsigma _2}}}{{x{\Omega _I}}}}}  - \frac{{{\varsigma _3}}}{{x{\Omega _I}}}\left[ {{\rm{Ei}}\left( { - \frac{{{\varsigma _2}}}{{x{\Omega _I}}}} \right) - {\rm{Ei}}\left( { - \frac{{{\varsigma _1}}}{{x{\Omega _I}}}} \right)} \right] \right\} \right]{\rm{dx}}  \nonumber \\  
& + \frac{1}{{\ln 2}}\int_{\rho {\varsigma _1}}^{\rho {\varsigma _2}} \frac{1}{{1 + x}}\left[  - \frac{{{d^2}}}{{R_n^2}}\left( {1 - {e^{ - \frac{{{\varsigma _2}}}{{x{\Omega _I}}} + \frac{1}{{{\Omega _I}\rho }}}}} \right) \right. \nonumber \\ 
& \left. + \frac{{{\varsigma _3}}}{{x{\Omega _I}}}{e^{\frac{1}{{{\Omega _I}\rho }}}}\left[ {{\rm{Ei}}\left( { - \frac{{{\varsigma _2}}}{{x{\Omega _I}}}} \right) - {\rm{Ei}}\left( { - \frac{1}{{{\Omega _I}\rho }}} \right)} \right] \right]{\rm{dx}}.
\end{align}
The integral above is rather involved. Let ${\varsigma _1} = \frac{{\varepsilon \eta {a_n}}}{{R_n^2 + {d^2}}}$, ${\varsigma _2} = \frac{{\varepsilon \eta {a_n}}}{{{d^2}}}$, ${\varsigma _3} = \frac{{\varepsilon \eta {a_n}}}{{R_n^2}}$, ${\upsilon _1}\left( {\rho ,x} \right) = \frac{1}{{1 + x}}[ 1 - {e^{\frac{1}{{{\Omega _I}\rho }}}}\{ \left( {1 + \frac{{{d^2}}}{{R_n^2}}} \right){e^{ - \frac{{{\varsigma _1}}}{{x{\Omega _I}}}}} - \frac{{{d^2}}}{{R_n^2}}{e^{ - \frac{{{\varsigma _2}}}{{x{\Omega _I}}}}} - \frac{{{\varsigma _3}}}{{x{\Omega _I}}} \left[ {{\rm{Ei}}\left( { - \frac{{{\varsigma _2}}}{{x{\Omega _I}}}} \right) - {\rm{Ei}}\left( { - \frac{{{\varsigma _1}}}{{x{\Omega _I}}}} \right)} \right] \} ]$, ${\upsilon _2}\left( {\rho ,x} \right) = \frac{1}{{1 + x}}[  - \frac{{{d^2}}}{{R_n^2}}\left( {1 - {e^{ - \frac{{{\varsigma _2}}}{{x{\Omega _I}}} + \frac{1}{{{\Omega _I}\rho }}}}} \right) + \frac{{{\varsigma _3}}}{{x{\Omega _I}}}{e^{\frac{1}{{{\Omega _I}\rho }}}}[ {\rm{Ei}}\left( { - \frac{{{\varsigma _2}}}{{x{\Omega _I}}}} \right) - {\rm{Ei}}\left( { - \frac{1}{{{\Omega _I}\rho }}} \right) ] ]$. After some integral transform methods, the original integral expression can be transformed into 
\begin{align}\label{Appendix C C_3}
\tilde R_n^{\rm{NISIC}} = \frac{1}{{\ln 2}}\int_{ - 1}^1 {\frac{1}{{\sqrt {1 - {t^2}} }}\left[ {{\Upsilon _1}\left( {\rho ,t} \right){\rm{ + }}{\Upsilon _2}\left( {\rho ,t} \right)} \right]{\rm{dt}}},
\end{align}
where ${\Upsilon _1}\left( {\rho ,t} \right) = \frac{{\rho {\varsigma _1}\sqrt {1 - {t^2}} }}{2}{\upsilon _1}\left( {\rho ,\frac{{\rho {\varsigma _1}}}{2}\left( {t + 1} \right)} \right)$, ${\Upsilon _2}\left( {\rho ,t} \right) = \frac{{\rho \left( {{\varsigma _2} - {\varsigma _1}} \right)\sqrt {1 - {t^2}} }}{2}{\upsilon _2}\left( {\rho ,\frac{{\rho \left( {{\varsigma _2} - {\varsigma _1}} \right)}}{2}\left( {t + 1} \right) + \rho {\varsigma _1}} \right)$. By using the Gauss-Chebyshev quadrature \cite[Eq. (8.8.12)]{Hildebrand1987introduction}, the final ergodic data rate of node $n$ with NISIC can be obtained as \eqref{ER of user n with ipSIC} upon plugging $\upsilon _1 \left( {\rho ,x} \right)$, $\upsilon _2 \left( {\rho ,x} \right)$, $\Upsilon _1 \left( {\rho ,t} \right)$, and $\Upsilon _2 \left( {\rho ,t} \right)$ into \eqref{Appendix C C_3}. The proof is completed.

\appendices
\section*{Appendix~D} \label{Appendix:D}
\renewcommand{\theequation}{D.\arabic{equation}}
\setcounter{equation}{0}
The proof begins upon plugging \eqref{SINR n} into \eqref{Definition Ergodic Rate} under $\varpi = 0$ condition. Similarly to Appendix C, it is easily observed that the CDF is almost identical to the outage probability expression of node $n$ under ISIC, and the only difference is that $\gamma_{thn}$ in the outage probability expression is replaced with $x$. Therefore, the ergodic data rate expression can be given by 
\begin{align}
&\tilde R_n^{\rm{ISIC}} = \frac{1}{{\ln 2}}\left[ \int_0^{\frac{{\varepsilon \eta \rho {a_n}}}{{R_n^2 + {d^2}}}} {\frac{{1 - 0}}{{1 + x}}{\rm{dx}}} \right. \nonumber \\ 
& \left. + \int_{\frac{{\varepsilon \eta \rho {a_n}}}{{R_n^2 + {d^2}}}}^{\frac{{\varepsilon \eta \rho {a_n}}}{{{d^2}}}} {\frac{{1 - \left( {1 - \frac{{\frac{{\varepsilon \eta \rho {a_n}}}{x} - {d^2}}}{{R_n^2}}} \right)}}{{1 + x}}{\rm{dx}}}  + \int_{\frac{{\varepsilon \eta \rho {a_n}}}{{{d^2}}}}^\infty  {\frac{{1 - 1}}{{1 + x}}{\rm{dx}}}  \right].
\end{align}
Let ${\chi _{n1}} = \frac{{\varepsilon \eta \rho {a_n}}}{{R_n^2 + {d^2}}},{\chi _{n2}} = \frac{{\varepsilon \eta \rho {a_n}}}{{{d^2}}}$. After a series of simple integral operations, the ergodic data rate expression is transferred to 
\begin{align}\label{Appendix D D_2}
\tilde R_n^{\rm{ISIC}} =& \frac{1}{{\ln 2}}\left\{ \ln \left( {1 + {\chi _{n1}}} \right) + \frac{1}{{R_n^2}}\left[ \varepsilon \eta \rho {a_n}\left( \ln \frac{{{\chi _{n2}}}}{{{\chi _{n1}}}} \right. \right. \right. \nonumber \\ 
& \left. \left. \left. - \ln \frac{{1 + {\chi _{n2}}}}{{1 + {\chi _{n1}}}} \right) - {d^2}\ln \frac{{1 + {\chi _{n2}}}}{{1 + {\chi _{n1}}}} \right] \right\}.
\end{align}

Following some basic mathematical transformations with plugging ${\chi _{n1}}$ and ${\chi _{n2}}$ into \eqref{Appendix D D_2}, the ergodic data rate of node $n$ under ISIC is obtained as \eqref{ER of user n with pSIC}. The proof is completed.

\appendices
\section*{Appendix~E} \label{Appendix:E}
\renewcommand{\theequation}{E.\arabic{equation}}
\setcounter{equation}{0}
The proof can begin upon plugging \eqref{SINR f} into \eqref{Definition Ergodic Rate}. It is directly observed that the CDF is very close to the outage expression of node $f$ with the only difference that the $\gamma_{thf}$ is substituted by $x$. As a result, the initial ergodic data rate expression of node $f$ can be initially expanded as
\begin{align}\label{Appendix E E_2}
{{\tilde R}_f}\left( \rho  \right) =& \frac{1}{{\ln 2}} \int_0^{{\chi _{f1}}} {\frac{{1 - 0}}{{1 + x}}{\rm{dx}}}  +  \frac{1}{{\ln 2}} \int_{{\chi _{f1}}}^{{\chi _{f2}}} \frac{1}{{1 + x}} \nonumber \\ 
& \times \left[ {1 - \left( {1 + \frac{{\varepsilon \eta \rho {a_n}}+ d^2}{{R_f^2}} - \frac{{\varepsilon \eta \rho {a_f}}}{{R_f^2}}\frac{1}{x}} \right)} \right]{\rm{dx}} \nonumber \\  
&  +  \frac{1}{{\ln 2}} \int_{{\chi _{f2}}}^\infty  {\frac{{1 - 1}}{{1 + x}}{\rm{dx}}}.
\end{align} 
The three integrals above are elementary, and letting ${\chi _{f1}} = \frac{{\varepsilon \eta \rho {a_f}}}{{R_f^2 + {d^2}}},{\chi _{f2}} = \frac{{\varepsilon \eta \rho {a_f}}}{{{d^2}}}$. After a series of transformations and basic calculations with plugging ${\chi _{f1}}$ and ${\chi _{f2}}$ into \eqref{Appendix E E_2}, the result is obtained as the ergodic data rate of node $f$ in PASS-NOMA networks by \eqref{ER of user f}. The proof is completed.

\appendices
\section*{Appendix~F} \label{Appendix:F}
\renewcommand{\theequation}{F.\arabic{equation}}
\setcounter{equation}{0}
Upon plugging \eqref{SINR f} into \eqref{Definition Ergodic Rate}, the proof starts. It is worth noting that the wireless channel between pinching antenna and node $f$ is considered as NLoS channel following Rayleigh distribution. Similar to the previous appendix, the CDF of node $f$'s SINR with NLoS condition has a similar structure, and the only difference is that a symbol $\gamma_{thf}$ in the formula is replaced with $x$. The integral expression of ergodic data rate for node $f$ under NLoS condition is represented by
\begin{align}
&\tilde R_f^{\rm{NLoS}}\left( \rho  \right) = \frac{1}{{\ln 2}}\int_0^\infty  \frac{1}{{1 + x}} \frac{{{\Omega _f}\rho \varepsilon \eta ({\textstyle{{{a_f}} \over x}} - {a_n})}}{{R_f^2}} \nonumber \\ 
& \times \left( {{e^{ - \frac{{{d^2}}}{{{\Omega _f}\rho \varepsilon \eta ({\textstyle{{{a_f}} \over x}} - {a_n})}}}} - {e^{ - \frac{{R_f^2 + {d^2}}}{{{\Omega _f}\rho \varepsilon \eta ({\textstyle{{{a_f}} \over x}} - {a_n})}}}}} \right){\rm{dx}}.
\end{align}
Because of the existence of relation $\frac{{{a_f}}}{x} - {a_n} > 0$, and $x>0$, it can be obtained that the upper and lower limits of integration are zero and $\frac{{{a_f}}}{{{a_n}}}$, respectively. By using the Gauss-Chebyshev quadrature \cite[Eq. (8.8.12)]{Hildebrand1987introduction}, the final ergodic data rate of node $f$ with the NLoS condition can be obtained as \eqref{ER of user f with NLoS}. The proof is completed.

\bibliographystyle{IEEEtran}
\bibliography{mybib}

\end{document}